\documentclass[12pt]{iopart}
 
\usepackage{color}
\usepackage{amsfonts}
\usepackage{amsthm,amssymb}
\usepackage{graphicx} 
\usepackage{subeqn} 
\usepackage{iopams,slashbox}  
\newcommand{\eps}{\varepsilon}
\newcommand{\lm}{\lim_{\tau\rightarrow\infty}}
\newcommand{\lms}{\lim_{\tau\rightarrow\tau_M^-}}
\newcommand{\beq}{\begin{equation}}
\newcommand{\eeq}{\end{equation}}
\newcommand{\bu}{\bar{u}_}
\newcommand{\bea}{\begin{eqnarray}}
\newcommand{\eea}{\end{eqnarray}}
\newcommand{\beas}{\begin{eqnarray*}}
\newcommand{\eeas}{\end{eqnarray*}}
\newcommand{\refb}[1]{(\ref{#1})}
\newtheorem{theorem}{Theorem}[section]
\newtheorem{proposition}{Proposition}[section]
\newtheorem{lemma}{Lemma}[section]
\newtheorem{corollary}{Corollary}[section]

\begin{document}

\title[Collapse of a self-similar cylindrical scalar field with non-minimal coupling I: Solutions with a regular axis]{Collapse of a self-similar cylindrical scalar field with non-minimal coupling}

\author{Eoin Condron and Brien C. Nolan}

\address{School of Mathematical Sciences, Dublin City University, Glasnevin, Dublin 9, Ireland.}
\eads{ \mailto{eoin.condron4@mail.dcu.ie},\mailto{brien.nolan@dcu.ie}}

\begin{abstract}
We investigate self-similar scalar field solutions to the Einstein equations in whole cylinder symmetry. Imposing self-similarity on the spacetime gives rise to a set of single variable functions describing the metric. Furthermore, it is shown that the scalar field is dependent on a single unknown function of the same variable and that the scalar field potential has exponential form. The Einstein equations then take the form of a set of ODEs. Self-similarity also gives rise to a singularity at the scaling origin. We discuss the number of degrees of freedom at an arbitrary point and prove existence and uniqueness of a 2-parameter family of solutions with a regular axis. We discuss the evolution of these solutions away from the axis toward the past null cone of the singularity, determining the maximal interval of existence in each case. 
\end{abstract}

\section{Introduction.}
The issue of cosmic censorship is one of the principal outstanding questions in general relativity. The cosmic censorship hypothesis comes in two versions, the weak and the strong. The weak hypothesis asserts that generically, in the gravitational collapse in asymptotically flat spacetimes, any singularity formed is shielded by an event horizon: singularities are globally censored. The strong hypothesis states that generic collapse leads to a globally hyperbolic spacetime: singularities are locally censored. There are, however, numerous examples of spacetimes which exhibit naked singularity (NS) formation, and the ultimate aim of this work is to determine if the class of spacetimes which give the title to this paper is numbered among them. 
Self-similarity plays an important role in many theories of classical physics. In general relativity, Carr's self-similarity hypothesis asserts that under certain physical conditions, solutions naturally evolve to a self-similar form \cite{Carr}. There is a body of evidence that supports this hypothesis: see e.g. \cite{CC}. Thus self-similar solutions are highly relevant to the study of gravitational collapse. Harada and Jhingan recently extended the self-similarity hypothesis to cylindrical spacetimes using perturbation analysis of Einstein-Rosen waves \cite{HJ}. The possibility of self-similar solutions acting as end states to more general solutions in the context of cylindrical symmetry is also discussed in \cite{NHKM}. The principal motivation for studying self-similar solutions is that the assumption of self-similarity brings about a significant simplification, reducing the field equations to ODEs. Self-similar spherically symmetric spacetimes are now well understood (see e.g. \cite{CC} and references therein). We note in particular the results of Brady \cite{Brady} and Christodoulou \cite{Christ1} which demonstrated the occurrence of naked singularities in the self-similar, spherically symmetric collapse of a (minimally coupled) scalar field. Christodoulou \cite{Christ2} proceeded to show that the NS solutions in the class are unstable in the general class of spherically symmetric scalar field spacetimes. The existence of a class of NS solutions for a well-behaved matter field is always of interest as these may act as seeds for a more general (non-self-similar) class of NS solutions. If such a class is found, the relevant question changes to that of their stability. Thus it is of interest to determine whether or not NS arise in a given well-defined class of spacetimes.\par
As a departure from spherical symmetry and in an effort to elucidate non-spherical collapse, some work has been done on cylindrical symmetry, for example \cite{A&T}-\cite{HNN}. 
This paper is the first of two, whose overall aim is to add to this body of work with a rigorous analysis of self-similar, cylindrically symmetric spacetimes coupled to a non-linear scalar field. 
We assume self-similarity of the first kind, where the homothetic Killing vector field (HKVF) is orthogonal to cylinders of symmetry, and show that there exists a curvature singularity where the homothetic Killing vector is identically zero, known as the scaling origin $\mathcal{O}$. We study solutions with a regular axis \cite{Hayward}. The associated regular axis conditions give rise to initial values for the independent variables of the problem, allowing us to cast the Einstein-Scalar Field equations as an initial value problem. To determine whether this class of spacetimes exhibits NS formation, we aim to find the global structure of solutions and determine whether the the future null cone of the singularity, which we label $\mathcal{N}_+$, exists as part of the spacetime. In the case where $\mathcal{N}_+$ is part of the spacetime, it corresponds to an absolute Cauchy horizon.
The field equations have three singular points; along the axis, along the past null cone of the origin, labelled $\mathcal{N}_-$, and along $\mathcal{N}_+$. This gives a natural division of the problem into two stages; solutions between the axis and $\mathcal{N}_-$, called region \textbf{I}, and between $\mathcal{N}_-$ and $\mathcal{N}_+$, called region \textbf{II}. The content of this paper deals with region \textbf{I}. In a follow up paper we investigate the structure of solutions which extend beyond $\mathcal{N}_-$ \cite{paper2}. In the second paper we find that all of relevant solutions satisfy a strong cosmic censorship condition. \par
The layout of the paper is as follows. Section 2 gives the formulation of the field equations and initial data along the axis, and a proof that a singularity exists at the scaling origin. We also show that the minimally coupled-case is mathematically equivalent to the vacuum case. In the general case, the system has two parameters and a free initial datum, different ranges of which give different global structures. We also summarise the main results of the paper in this section, with the proofs given in later sections. 
In section 3 we prove existence and uniqueness of solutions to the initial value problem formulated in section 2, using a fixed point argument. In sections 4 and 5 we determine the global structure of solutions in region \textbf{I}. In some cases, exact solutions may be found and these are deferred to section 5. Section 6 gives the proof of the main theorem of this paper and we conclude with a brief summary in section 7. We use units such that $8\pi G=c=1$ throughout. 
\section{Self-similar cylindrically symmetric spacetimes coupled to a non-linear scalar field with a regular axis}
Here we give the Einstein equations and initial conditions that correspond to self-similar cylindrically symmetric scalar field spacetimes. In section 2.1, we give the line element for whole cylinder symmetry, and write down the Einstein equations for the case of a non-linear scalar field. In section 2.2, we treat the minimally coupled case, showing that it is mathematically equivalent to the vacuum case. In section 2.3 we specialise to the case of self-simliarity of the first kind, and show how this assumption reduces the EFEs to a set of ODEs. We also determine that the scalar field potential has exponential form. In section 2.4, we discuss the regular axis conditions, and show how these
give rise to initial conditions for the ODEs derived in section 2.3. We then prove in section 2.5 that, except in the case of flat spacetime, there is a singularity at the scaling origin. In section 2.6 we write down the initial value problem that is studied in the remainder of the paper and we conclude with a statement of the main results of the paper in section 2.7. 
\subsection{The Einstein equations for a cylindrically symmetric scalar field}
We consider cylindrically symmetric spacetimes with whole-cylinder symmetry \cite{A&T}. This class of spacetimes admits a pair of commuting, spatial Killing vectors $\boldsymbol{\xi}_{(\theta)}$,$\boldsymbol{\xi}_{(z)}$ called the axial and translational Killing vectors, respectively. Introducing double null coordinates $(u,v)$ on the Lorentzian 2-spaces orthogonal to the surfaces of cylindrical symmetry, the line element may be written as:
\begin{equation}
\label{LE uv}ds^2 = -2e^{2\bar{\gamma}+2\bar{\phi}}dudv + e^{2\bar{\phi}}r^2d\theta^2 + e^{-2\bar{\phi}}dz^2,
\end{equation}
where $r$ is the radius of cylinders, $\bar{\gamma} ,\phi$  and  $r$  depend on $u$ and $v$ only.\\
We take the matter source to be a cylindrically symmetric, self-interacting scalar field $\psi(u,v)$ with stress-energy tensor given by
\begin{equation} 
\label{EMT}T_{ab}= \nabla_{a}\psi\nabla_{b}\psi-\frac{1}{2}g_{ab}\nabla^{c}\psi\nabla_{c}\psi-g_{ab}V(\psi),
\end{equation}
where $V(\psi)$ is the scalar field potential.
The form of the line element is preserved by the coordinate transformations
\begin{equation}
\label{Coord Freedom}u \rightarrow \bar{u}(u),\quad v \rightarrow \bar{v}(v), \quad z \rightarrow \lambda z,
\end{equation}
for constant $\lambda$. Note that $\theta\in[0,2\pi)$ and so transformations of the kind $\theta\rightarrow \lambda \theta$ are not allowed in general. 
The full set of Einstein equations for these spacetimes is
\begin{subequations}
\begin{eqnarray}
\label{EFE uv a}2r_{u}\bar{\gamma}_{u}-r_{uu} -2r\bar{\phi}_{u}^2 = r\psi_{u}^2,\\
\label{EFE uv b}r_{uv}=   r e^{2\bar{\gamma}+2\bar{\phi}} V(\psi),\\
\label{EFE uv c}2r_{v}\bar{\gamma}_{v}-r_{vv} - 2r\bar{\phi}_{v}^2 = r\psi_{v}^2,\\
\label{EFE uv d}2(\bar{\phi}_u\bar{\phi}_v+\bar{\gamma}_{uv})= - \psi_u\psi_v+e^{2\bar{\gamma}+2\bar{\phi}} V(\psi)
\end{eqnarray}
\begin{equation}
\label{EFE uv e}2(r_u\bar{\phi}_v+r_v\bar{\phi}_u+r\bar{\phi}_u\bar{\phi}_v+r_{uv}+r\bar{\gamma}_{uv}+2r\bar{\phi}_{uv})=-r\psi_u\psi_v+re^{2\bar{\gamma}+2\bar{\phi}} V(\psi)
\end{equation}\end{subequations}
The subscripts denote partial derivatives here. Our field $\psi$ satisfies
\bea
\label{Scalar wave}\nabla^a\nabla_a\psi-V'(\psi)=0, 
\eea
which implies $\nabla^aT_{ab}=0$. This yields
\beq
\label{wave psi uv}2r\psi_{uv}+r_v\psi_u+r_u\psi_v+ r e^{2\bar{\gamma}+2\bar{\phi}}V'(\psi)=0,
\eeq
which is the wave equation for the scalar field $\psi$. A useful simplification is obtained by subtracting $r$\refb{EFE uv d} from \refb{EFE uv e}, which gives
\beq
\label{wave phi} 2r\bar{\phi}_{uv}+r_{u}\bar{\phi}_{v}+r_{v}\bar{\phi}_{u}+r_{uv}=0.\\
\eeq
Note that \refb{wave psi uv} can be derived from \refb{EFE uv a}-\refb{EFE uv e} and that from this point onward we make use of \refb{EFE uv a}-\refb{EFE uv e},\refb{wave psi uv} and \refb{wave phi} in our analysis. 
\subsection{The minimally coupled case}
In this section we deal with the case where the scalar field potential $V$ is equal to zero, and so the scalar field is minimally coupled ($\nabla_a\nabla^a\psi=0)$. The field equations simplify greatly in this case and we show that solving the Einstein equations is effectively the same as in the vacuum case. With $V=0$, equation \refb{EFE uv b} gives 
\beq
r_{uv}=0, \qquad r= f(u)+g(v).
\eeq
We require the absence of trapped cylinders in the initial configuration so the gradient of $r$ must be spacelike \cite{Kip}. This reduces to the condition
\beq
f'(u)g'(v)<0.
\eeq
Using the coordinate freedom \refb{Coord Freedom}, we then set 
\beq
r=\frac{v-u}{\sqrt{2}}. 
\eeq
To demonstrate equivalence to the vacuum case, we follow the example of \cite{HNN} and introduce time and radial coordinates 
\beq
T=\frac{v+u}{\sqrt{2}}, \qquad X=\frac{v-u}{\sqrt{2}}. 
\eeq
The line  element is then given by 
\beq
ds^2 = e^{2\bar{\gamma}+2\bar{\phi}}(dX^2-dT^2)+X^2e^{2\bar{\phi}}d\theta^2+e^{-2\bar{\phi}}dz^2,
\eeq
and the remaining field equations reduce to 
\begin{subequations}
\bea
\label{gamma X}\bar{\gamma}_X=X\left(\bar{\phi}_T^2+\bar{\phi}_X^2+\frac{\psi_T^2}{2}+\frac{\psi_X^2}{2}\right),\\
\label{gamma T}\bar{\gamma}_T=X(2\bar{\phi}_T\bar{\phi}_X+\psi_X\psi_T),\\
\psi_{TT}-\psi_{XX}-\frac{\psi_X}{X}=0,\\
\bar{\phi}_{TT}-\bar{\phi}_{XX}-\frac{\bar{\phi}_X}{X}=0.\eea
\end{subequations}
Given regular initial data, the linear wave equations for $\psi$ and $\bar{\phi}$ yield unique, globally hyperbolic, singularity-free solutions.  Solutions of $\bar{\gamma}$ may be then be obtained from \refb{gamma X} and \refb{gamma T}. We note that this is, essentially, mathematically equivalent to the vacuum case and refer the reader to \cite{ABS} and \cite{HNN} (for the self-similar case) for a full treatment of the problem. 
\subsection{Self-similarity}
We assume self-similarity of the first kind \cite{Carr} which is equivalent to the existence of a HKVF $\boldsymbol{\boldsymbol{\xi}}$, such that
\begin{equation}
\label{SS}\mathcal{L}_{\boldsymbol{\xi}} g_{ab} = 2 g_{ab},
\end{equation}
where $\mathcal{L}_{\boldsymbol{\xi}}$ denotes the Lie derivative along the vector $\boldsymbol{\boldsymbol{\xi}}$. We consider only cylindrically symmetric HKVF's. 
They have the form
\begin{equation}
\label{HKV}\boldsymbol{\boldsymbol{\xi}} = \alpha(u,v)\frac{\partial}{\partial u}+\beta(u,v)\frac{\partial}{\partial v},
\end{equation}
so that $\boldsymbol{\boldsymbol{\xi}}$ is orthogonal to the the cylinders of symmetry. We note that $\boldsymbol{\xi}$ could have $\partial_\theta,\partial_z$ components and that non-axis-orthogonal HKVF's may give a richer structure to the space of solutions. In \cite{HNN}, the authors consider self-similar cylindrical vacuum solutions. When the HKVF is assumed to be cylindrical, they find that the spacetime is actually flat. They then consider more general one-parameter family of HKVF's $\boldsymbol{w}$ of the form 
\beq
\boldsymbol{w} = \frac{1}{1-\kappa}\frac{\partial}{\partial t} + \frac{1}{1-\kappa}\frac{\partial}{\partial x} +\frac{1-2\kappa}{1-\kappa}\frac{\partial}{\partial z},
\eeq
where $t$ and $x$ are time and radial coordinates. 
Depending on the choice of the parameter $\kappa$, these spacetimes are found to describe the interior of an
exploding (imploding) shell of gravitational waves or the collapse (explosion) of gravitational waves.
However, we have chosen to restrict our study to HKVF of the form \refb{HKV} for the sake of simplicity. We note that this gives a well-defined class of HKVF's. Equation
\refb{SS} is equivalent to 
\beq
\nabla_{\mu}\boldsymbol{\xi}_{\nu}+\nabla_{\nu}\boldsymbol{\xi}_{\mu}= 2g_{\mu \nu},\eeq
which leads to $\alpha = \alpha(u)$ and $\beta=\beta(v)$. 
We then use the coordinate freedom \refb{Coord Freedom} to rescale $u$ and $v$ such that $\alpha(u) = 2u$ and $\beta(v) = 2v$.\\
Having made this transformation, equations \refb{SS}
 yield
\begin{equation}
\label{Single Var}\bar{\gamma} =\gamma(\eta), \qquad \bar{\phi}=\phi(\eta)-\log|u|^{1/2},\qquad r =|u|S(\eta),
\end{equation}
where 
\beq
\eta = \frac{v}{u},\eeq
is the similarity variable and $\gamma,\phi,S$ are metric functions for the self-similar metric, which is given by
\begin{equation}
\label{LE eta}ds^2 = -2|u|^{-1}e^{2\gamma(\eta)+ 2\phi(\eta)}dudv+|u|e^{2\phi(\eta)}S^2(\eta)d\theta+|u|e^{-2\phi(\eta)}dz^2.
\end{equation}
The coordinate transformations that preserve this form of the metric are 
\begin{equation}
\label{Coord Freedom 2}u \rightarrow \lambda u,\quad v \rightarrow\mu v, \quad z \rightarrow \sigma z,
\end{equation}
for constants $\lambda,\mu,\sigma$. The following result is stated without proof in \cite{KHM} and \cite{Ellis}, however, we found it useful to give a proof here. 
\begin{proposition} For a self-similar scalar field $\psi$ with energy-momentum tensor \refb{EMT} and $V(\psi)\neq0$, admitting a homothetic Killing vector $\boldsymbol{\xi}$ such that \refb{SS} holds, the potential $V(\psi)$ has the exponential form 
\beq 
\label{V} V(\psi)=\bar{V}_0e^{-2\psi/k},
\eeq
where $\bar{V}_0\neq0 ,k\neq0$ are constants. \end{proposition}
\begin{proof}
It can be shown that \refb{SS} leads to $\mathcal{L}_{\boldsymbol{\xi}}T_{ab}=0$ \cite{Carr}. For $T_{ab}$ given by \refb{EMT} we have
\begin{eqnarray}
\nonumber\psi_a \mathcal{L}_{\boldsymbol{\xi}}\psi_b+\psi_b \mathcal{L}_{\boldsymbol{\xi}}\psi_a-g_{ab}\bigg(\psi^c\psi_c+\frac{1}{2}\psi^c\mathcal{L}_{\boldsymbol{\xi}}\psi_c
\label{SS EMT}+\frac{1}{2}\psi_c \mathcal{L}_{\boldsymbol{\xi}}\psi^c\\ \hspace{150pt}+2V+V'(\psi) \mathcal{L}_{\boldsymbol{\xi}}\psi\bigg)=0.
\end{eqnarray}
Now $\mathcal{L}_{\boldsymbol{\xi}}\psi_c=\mathcal{L}_{\boldsymbol{\xi}}g_{bc}\psi^b=2\psi_c+g_{bc}\mathcal{L}_{\boldsymbol{\xi}}\psi^b$, and so 
\begin{equation}
\label{switch}\psi^c\mathcal{L}_{\boldsymbol{\xi}}\psi_c=2\psi^c\psi_c+\psi_c \mathcal{L}_{\boldsymbol{\xi}}\psi^c.
\end{equation}
Combining this with \refb{SS EMT} and taking the trace then yields
\begin{equation}
\label{trace 1}-\psi^c \mathcal{L}_{\boldsymbol{\xi}}\psi_c=4V+2V'(\psi) \mathcal{L}_{\boldsymbol{\xi}}\psi.
\end{equation}
Using \refb{trace 1} to eliminate $2V+V'(\psi) \mathcal{L}_{\boldsymbol{\xi}}\psi$ from \refb{SS EMT} and simplifying produces
\begin{equation}
\label{tensor ab}\psi_a \mathcal{L}_{\boldsymbol{\xi}}\psi_b+\psi_b \mathcal{L}_{\boldsymbol{\xi}}\psi_a-\frac{1}{2}g_{ab}\psi^c\mathcal{L}_{\boldsymbol{\xi}}\psi_c=0.
\end{equation}
Contracting with $\psi^a$ gives
\begin{equation}
\label{tensor b}\psi^c\psi_c \mathcal{L}_{\boldsymbol{\xi}}\psi_b+\frac{1}{2}\psi_b\psi^c \mathcal{L}_{\boldsymbol{\xi}}\psi_c=0,
\end{equation}
and contracting with $\psi^b$ gives
\begin{equation}
\label{scalar}\frac{3}{2}\psi^b\psi_b(\psi^c \mathcal{L}_{\boldsymbol{\xi}}\psi_c)=0.
\end{equation}
In the case $\psi^c\psi_c=0$ we have $\psi^c\mathcal{L}_{\boldsymbol{\xi}}\psi_c=0$, from \refb{tensor b}, since we are assuming $\psi_b\neq0$. It then follows from \refb{switch} that $\psi_c\mathcal{L}_{\boldsymbol{\xi}}\psi^c=0$. Contracting \refb{tensor ab} with $\mathcal{L}_{\boldsymbol{\xi}}\psi^b$ produces
\begin{equation}
\label{tensor a 2}\psi_a  \mathcal{L}_{\boldsymbol{\xi}}\psi^b\mathcal{L}_{\boldsymbol{\xi}}\psi_b+\psi_b\mathcal{L}_{\boldsymbol{\xi}}\psi^b\mathcal{L}_{\boldsymbol{\xi}}\psi_a=\psi_a  \mathcal{L}_{\boldsymbol{\xi}}\psi^b\mathcal{L}_{\boldsymbol{\xi}}\psi_b=0,
\end{equation}
and we see that $\mathcal{L}_{\boldsymbol{\xi}}\psi_b$ is null. Since it is also orthogonal to $\psi_b$, it must be parallel to it, i.e., $\mathcal{L}_{\boldsymbol{\xi}}\psi_b=\Lambda\psi_b$ for some quantity $\Lambda$. Putting this into \refb{tensor ab} gives $2\Lambda\psi_a\psi_b=0$, which reveals that $\Lambda$ must be zero, i.e. $\mathcal{L}_{\boldsymbol{\xi}}\psi_b=0$.\\
In the case $\psi^c\psi_c\neq 0$, we also have $\psi^c \mathcal{L}_{\boldsymbol{\xi}}\psi_c=0$, by \refb{scalar}. It follows immediately from \refb{tensor b} that $\mathcal{L}_{\boldsymbol{\xi}}\psi_b=0$ in this case also.
It is straightforward to show that $\partial_b\mathcal{L}_{\boldsymbol{\xi}}\psi=\mathcal{L}_{\boldsymbol{\xi}}\psi_b$, so we have $\partial_b\mathcal{L}_{\boldsymbol{\xi}}\psi=0$, and thus $\mathcal{L}_{\boldsymbol{\xi}}\psi=k$,
for some constant $k$. Equation \refb{trace 1} then simplifies to $2V+kV'=0$, which yields \refb{V} for $k\neq0$. Note that $k=0$ gives $V=0$, which has been dealt with above.
\end{proof} 

\begin{corollary}
If $\boldsymbol{\xi}$ has the form \refb{HKV} with $\alpha=2u$ and $\beta =2v$, then $\psi$ and $V(\psi)$ may be written as 
\begin{equation}
\label{SS SF}\psi=F(\eta)+\frac{k}{2}\log|u|, \qquad V(\psi) = \frac{\bar{V}_0e^{-\frac{2}{k}F(\eta)}}{|u|}.
\end{equation}
\end{corollary}
\begin{proof}
In this case, $\mathcal{L}_{\boldsymbol{\xi}}\psi=k$ reduces to 
\begin{equation}
\mathcal{L}_{\boldsymbol{\xi}}\psi=2u\psi_u+2v\psi_v=k,
\end{equation}
from which $\psi=F(\eta)+\log|u|^{k/2}$ follows. Theorem 2.1 then gives the potential $V$. 
\end{proof}
We are now in a position to formulate the field equations as a set of ODEs.
In terms of $\gamma,\phi,S$ and $F$, \refb{EFE uv a}-\refb{EFE uv c},\refb{wave psi uv} and \refb{wave phi} are given by 
\begin{subequations}
\begin{eqnarray}
\label{EFE eta A}2\eta\gamma'(S-\eta S')+\eta^2S''+2S\left(\eta\phi'+\frac{1}{2}\right)^2=-S\left(\eta F'-\frac{k}{2}\right)^2,\\
\label{EFE eta B}\eta S''=-\bar{V}_0Se^{2\gamma+2\phi-2F/k},\\
\label{EFE eta C}2S'\gamma'-S''-2S\phi'^2=SF'^2,\\
\label{EFE eta D}2\eta S''+4\eta S\phi''+4\eta S'\phi'+2S\phi'+S'=0,\\
\label{EFE eta E}2\eta SF''+2\eta S'F'+SF'-\frac{kS'}{2}+\frac{2\bar{V}_0}{k}Se^{2\gamma+2\phi-2F/k}=0.
\end{eqnarray}
\end{subequations}
Now, \refb{EFE eta A}+$\eta^2$\refb{EFE eta C} simplifies to 
\begin{equation}
\label{GPF' eta}\frac{1}{2}+2\eta\gamma'+2\eta\phi'=k\eta F' -\frac{k^2}{4}.\end{equation}
Dividing by $\eta$ and integrating gives
\begin{equation}
\label{GPF eta}2\gamma+2\phi=k F -\left(\frac{1}{2}+\frac{k^2}{4}\right)\log|\eta|+c_1,\end{equation}
for some constant $c_1$. Equation \refb{EFE eta B} then reduces to
\begin{equation}
\label{S'' F}\eta S''=V_0e^{(k-2/k)F}|\eta|^{-(1/2+k^2/4)}S,
\end{equation}
where $V_0=\bar{V}_0e^{c_1}$ and we have used \refb{GPF eta} to replace $e^{2\gamma+2\phi}$. We define 
\begin{equation}
\label{l, lambda}l=\frac{2F}{k}-\log|\eta|^{1/2},\qquad \lambda = \frac{k^2}{2}-1,
\end{equation}
which gives
\begin{equation}
\label{S'' l}\eta S''=-V_0|\eta|^{-1}e^{\lambda l}S.
\end{equation}
Equation \refb{EFE eta D} is exact and may be integrated to give
\begin{equation}
\label{S' phi}2S\phi'+S'=c_2|\eta|^{-1/2},
\end{equation}
for some constant $c_2$. Written in terms of $l$ and $S$, \refb{EFE eta E} becomes 
\begin{equation}
\label{Sl''}\eta Sl'' +\eta S'l' + \frac{Sl'}{2}-\frac{S}{4\eta}+ \frac{2V_0}{k^2|\eta|}Se^{\lambda l}=0.
\end{equation}
\subsection{The regular axis conditions}
To ensure that the collapse ensues from an initially regular configuration we impose regular axis conditions \cite{Hayward} to the past of the scaling origin $(u,v)=(0,0)$. 
The areal radius $\rho$ and the specific length $L$ of the cylinders are given by the norms of the Killing vectors:
\begin{equation}\rho = \sqrt{\xi_{(\theta)}^a.\xi_{(\theta)a}}=|u|^\frac{1}{2}e^\phi S,  \qquad L= \sqrt{\boldsymbol{\xi}_{(z)}^a.\boldsymbol{\xi}_{(z)a}}=|u|^{\frac{1}{2}}e^{-\phi}.
\end{equation}
The axis is defined by $\rho = 0$. We rule out the case $u = 0$ as this is a null hypersurface and we require the axis to be timelike. For a regular axis, the specific length $L$ must be non zero and finite, and so $\phi$ must be also be finite. A regular axis must therefore correspond to $S(\eta) = 0$. Hence, $\eta$ must be constant along the axis and a rescaling of $u$ and $v$ using the coordinate freedom \refb{Coord Freedom 2} places the axis at $\eta = 1$.\\
Note that the past null cone of the origin $\mathcal{N}_-$ corresponds to $\eta =0$ and the interval $\eta\in[0,1]$ constitutes region \textbf{I}. \\
Further conditions for a regular axis are as follows \cite{Hayward}:
\begin{equation}
\label{Reg Axis}\nabla^a\rho \nabla_a\rho = 1 + O(\rho^2),\quad \nabla^a\rho \nabla_a L = O(\rho),\quad \nabla^a L\nabla_a L = O(1),
\end{equation}
where the big-oh relations hold in the limit $\rho \rightarrow 0$. 
The first condition ensures the standard $2\pi$-periodicity of the azimuthal coordinate near the axis, while the remaining conditions ensure the absence of any curvature singularities at the axis. 
\begin{proposition} The regular axis conditions reduce to the following set of data:
\begin{eqnarray}
\label{initial data} S(1)=0,\quad S'(1)=-1,\quad\phi'(1)=-\frac{1}{4},\quad \gamma'(1)=0, \quad l'(1)=0.
\end{eqnarray}
\end{proposition}
\begin{proof}
Note that $u<0$ to the past of the origin (0,0) and, therefore, $u<0$ on the axis. 
The equations \refb{Reg Axis} then give
\begin{subequations}\begin{eqnarray}
\label{lim 1}\lim_{\eta\rightarrow1}2e^{-2\gamma}(S'+S\phi')^2 = 1,\\
\label{lim 2}\lim_{\eta\rightarrow1}e^{-2\gamma}\left(S'+S\phi'\right)\left(\frac{1}{2}+2\phi'\right)=0,\\
\label{lim 3}\lim_{\eta\rightarrow1}2e^{-2\gamma}\left(\frac{1}{2}+\phi'\right)\phi'=L_0,
\end{eqnarray}
for some $L_0\in\mathbb{R}$. Equation \refb{lim 1} gives 
\beq
\label{lim 4}\lim_{\eta\rightarrow1}e^{-\gamma}(S'+S\phi') = \pm\frac{1}{\sqrt{2}},\eeq
which may be used to simplify \refb{lim 2} to 
\beq
\label{lim 5} \lim_{\eta\rightarrow1}\frac{e^{-\gamma}}{\sqrt{2}}\left(\frac{1}{2}+2\phi'\right)=0.\eeq
So we either have $\lim_{\eta\rightarrow1}\phi'=-1/4$ or $\lim_{\eta\rightarrow1}e^{-\gamma}=0$. In the latter case, we must also have $\lim_{\eta\rightarrow1}e^{-\gamma}\phi'=0$.
Now \refb{S' phi} may be used to replace $S'+S\phi'$ with $c_2-S\phi'$ in \refb{lim 4}, and so 
\beq
\label{lim 6}\lim_{\eta\rightarrow1}e^{-\gamma}(c_2-S\phi') = \pm\frac{1}{\sqrt{2}}.\eeq\end{subequations}
This clearly contradicts $\lim_{\eta\rightarrow1}e^{-\gamma}=\lim_{\eta\rightarrow1}e^{-\gamma}\phi'=0$ since $\lim_{\eta\rightarrow1}S=0$, so we must have $\phi'(1)=-1/4, S'(1)=c_2$ and $e^{-\gamma(1)}=\pm\sqrt{2}c_2\neq0$. \\
The circumferential radius $\rho$, and therefore $S$, must increase away from the axis. Recall that $\eta\in[0,1]$ in region \textbf{I}, so that $\eta$ is decreasing away from the axis at $\eta=1$. We must then have $S'(1)<0$. Note that the field equations (26) are invariant under the transformation $\bar{S}\rightarrow S/(-S'(1))$, so we may set $S'(1)=-1$.\\
It follows from finiteness of $\gamma(1),\phi(1)$ and equations \refb{GPF eta},\refb{l, lambda} that $F(1)$ and $l(1)$ must also be finite. Equation \refb{S'' l} then gives $S''(1)=0$ and, using this fact,  \refb{EFE eta C} gives $\gamma'(1)=0$. Inserting $\phi'(1)=-1/4$ and $\gamma'(1)=0$ into \refb{GPF' eta}, we arrive at $F'(1)=k/4$, which is equivalent to $l'(1)=0$. 
\end{proof}
\begin{proposition} In the case $\psi^c\psi_c=0$, solutions to the Einstein equation with line element and energy-momentum tensor given by \refb{LE eta} and \refb{EMT}, respectively, admit a regular axis if and only if $k=0$.\end{proposition}
\begin{proof} First note that $\psi^c\psi_c=2g^{01}\psi_u\psi_v=0$ leads to either $\psi_u=0$ or $\psi_v=0$. Equation \refb{SS SF} then gives
\begin{subequations}\begin{eqnarray}
\psi_u&=-\frac{vF'}{u^2}+\frac{k}{2u}=0,\qquad \psi_v = \frac{F'}{u}=0,\\
&\Rightarrow F'=\frac{k}{2\eta},\hspace{78pt}\Rightarrow F'=0.
\end{eqnarray}\end{subequations}
In both cases, $F'(1)=k/4$ holds if and only if $k=0$. \end{proof}
Recall that $k=0$ gives $V=0$, which is the minimally-coupled case dealt with in section 2.2.  
\subsection{Singular nature of the scaling origin}
As an immediate consequence of the assumption of self-similarity, in the non-minimally coupled case, there exists a spacetime singularity where the homothetic Killing vector is identically zero, i.e., at $(u,v)=(0,0)$. 
\begin{proposition} Let $\mathcal{T}$ denote the scalar invariant $T^{ab}T_{ab}$. Then 
\beq\lim_{u\rightarrow0}\mathcal{T}\,\bigg|_{\eta=1}=\infty.\eeq \end{proposition}
\begin{proof} 
It is straight forward to show that  
\bea
\mathcal{T}&=2(g^{uv})^2\left(T_{uu}T_{vv}+T_{uv}^2\right)+\left(g^{\theta\theta}T_{\theta\theta}\right)^2+\left(g^{zz}T_{zz}\right)^2\\
\nonumber&\ge 2(g^{uv})^2T_{uu}T_{vv}. \eea
Now,
\beq(g^{uv})^2T_{uu}T_{vv}=\frac{e^{-4\gamma-4\phi}}{u^2}\psi_u^2\psi_v^2=e^{-4\gamma-4\phi}\left(\eta F'-\frac{k}{2}\right)^2\frac{F'^2}{u^2}.\eeq
Using $F'(1)=k/4$ we have
\beq\mathcal{T}(1)\ge 2e^{-4\gamma(1)-4\phi(1)}\left(\frac{k^2}{16u}\right)^2.\eeq
Since $\gamma(1),\phi(1)$ are finite and $k\neq0$, taking the limit $u\rightarrow0$ completes the proof. \end{proof}
\subsection{The initial value problem}
Here we derive the form of the Einstein equations with which we will work for the remainder of the paper. The following choice of dependent and independent variables gives an autonomous system. 
\begin{proposition} [Autonomous field equations] Let 
\beq 
\tau=-\log\eta \qquad R=e^{\tau/2}S.
\eeq
Then the interval $\eta\in[1,0)$, i.e. region \textbf{I} of the spacetime, corresponds to $\tau\in[0,\infty)$ and the field equations are equivalent to 
\begin{subequations}\begin{eqnarray}
\label{EFE tau A}2\gamma+2\phi=\frac{k^2l}{2} + \frac{\tau}{2}+c_1,\\
\label{EFE tau B}\ddot{R}  = \left(\frac{1}{4} - V_0 e^{\lambda l}\right)R,\\
\label{EFE tau C}\dot{R}+\left(2\dot{\phi}-\frac{1}{2}\right)R= 1,\\
\label{EFE tau D}R\ddot{l}+ \dot{R}\dot{l}=\left(\frac{1}{4} - \frac{2}{k^2}V_0 e^{\lambda l}\right)R,\\
\label{EFE tau E}\frac{\dot{R}^2-1}{R^2}+\frac{k^2\dot{R}\dot{l}}{R}+2V_0e^{\lambda l} -\frac{2+k^2}{8}-\frac{k^2\dot{l}^2}{2}=0,\\
\label{EFE tau F}R(0) = 0,\quad \dot{R}(0)=1, \quad l(0) = l_0,\quad\dot{l}(0) = 0,
\end{eqnarray}\end{subequations}
where the overdot denotes a derivative with respect to $\tau$.\end{proposition}
\begin{proof}
For a general function $f$ we have $\eta f'=-\dot{f} $ and $\eta^2f''=\ddot{f}+\dot{f}$. Then for $\eta>0$, multiplying equations \refb{S'' l},\refb{S' phi},\refb{Sl''} by $\eta$ and changing variables gives
\begin{subequations}
\bea
\ddot{S}+\dot{S}= -V_0e^{\lambda l}S,\\
\label{Sphidot}2S\dot{\phi}+\dot{S}=-c_2e^{-\tau/2}=e^{-\tau/2},\\
S\ddot{l}+\dot{S}\dot{l}+\frac{S\dot{l}}{2}=S\left(\frac{1}{4}-\frac{2V_0e^{\lambda l}}{k^2}\right).
\eea\
We also have 
\bea
\dot{S}=e^{-\tau/2}\left(\dot{R}-\frac{R}{2}\right),\\
\ddot{S}=e^{-\tau/2}\left(\ddot{R}-\dot{R}+\frac{R}{4}\right),
\eea\end{subequations}
which, when combined with the above, give \refb{EFE tau B}-\refb{EFE tau D}. Equation \refb{EFE tau A} comes directly from \refb{GPF eta},\refb{l, lambda} and the definition of $\tau$. Multiplying \refb{EFE eta C} by $\eta^2$ and changing variables yields
\beq
2\left(\dot{R}-\frac{R}{2}\right)\dot{\gamma}+V_0e^{\lambda l}R-2R\dot{\phi}^2=\frac{k^2R}{4}\left(\dot{l}-\frac{1}{2}\right)^2,
\eeq
where we have used \refb{l, lambda} and \refb{S'' l} to replace $\eta F'$ and $\eta^2S''$, respectively. Using \refb{EFE tau C} and the derivative of \refb{EFE tau A} to replace $\dot{\gamma}$ and $\dot{\phi}$ with expressions in $R$ and $l$ then produces
\bea
\nonumber\left(\dot{R}-\frac{R}{2}\right)\left(\frac{k^2\dot{l}}{2}-\frac{1}{R}+\frac{\dot{R}}{R}\right)+V_0e^{\lambda l}R\\
\hspace{120pt}-\frac{R}{2}\left(\frac{1}{R}-\frac{\dot{R}}{R}+\frac{1}{2}\right)^2=\frac{k^2R}{4}\left(\dot{l}-\frac{1}{2}\right)^2.
\eea
Multiplying by 2/R and simplifying, we arrive at \refb{EFE tau E}. Finally, the axis is at $\tau=0$, so $R(0)=0,l(0)=l_0$.
Furthermore, $\dot{l}(0)=-l'(1)=0$ and $\dot{S}(0)=-S'(1)=\dot{R}(0)-R(0)/2$ which gives $\dot{R}(0)=1$. \end{proof}
Equations (50) are equivalent to the Einstein field equations with line element \refb{LE eta}, energy-momentum tensor \refb{EMT} and a regular axis. 
Henceforth, we study equations \refb{EFE tau B},\refb{EFE tau D} and \refb{EFE tau E} to determine solutions for $R$ and $l$. $\phi$ is then obtained by integrating \refb{EFE tau C} and once this is found, $\gamma$ is given by \refb{EFE tau A}. Note that $\tau\in[0,\infty)$ in region \textbf{I} and $\tau\rightarrow\infty$ at $\mathcal{N}_-$. 
\subsection{Summary of results} In this section we present two theorems which summarise the main results of the paper, whose proofs will follow in the following three sections. 
\begin{theorem} [Existence and Uniqueness]
Let $k,V_0\in \mathbb{R}$. For each $l_0,\phi_0\in \mathbb{R}$ there exists a unique solution of (50) on an interval $[0,\tau_*]$ corresponding to a spacetime with line element \refb{LE eta} and energy-momentum tensor \refb{EMT}. The spacetime admits a regular axis for $u+v<0$.\end{theorem}
\begin{theorem} [Global structure of solutions in the causal past of the scaling origin $\mathcal{O}$] For each $k,V_0,l_0\in\mathbb{R}$, let $[0,\tau_M)$ be the maximal interval of existence for the unique solution of (50). The global structure of the solution is given by one of the following cases: \\
\\
Case 1. If $k^2>2,V_0<0$, then $\tau_M<\infty$ and
the hypersurface $\tau=\tau_M$ corresponds to radial null infinity, with the Ricci scalar decaying to zero there.\\
\\
Case 2. If $k^2=2$ and $V_0<0$, then $\tau_M=+\infty$ and $\mathcal{N}_-$ corresponds to radial null infinity, with the Ricci scalar decaying to zero there. \\
\\
Case 3. If 
\beas (i)\quad &k^2<2\quad  and \quad V_0<0,\\
(ii)\quad &k^2>2, V_0>0 \quad and \quad V_0e^{\lambda l_0}<k^2/8,\eeas
 then $\tau_M=\infty$,
the radius of the cylinders is non zero and finite on $\mathcal{N}_-$, which is reached in finite affine time along radial null rays. Hence, $\mathcal{N}_-$ exists as part of the spacetime.\\
\\
Case 4. If
\beas
(i)\quad &k^2<2,V_0>0\quad and \quad V_0e^{\lambda l_0}> k^2/8,\\
(ii)\quad &k^2>2 \quad and \quad V_0>0,\\
(iii)\quad &k^2=2\quad  and \quad V_0>1/4,\eeas
then $\tau_M<\infty$ and there is a singularity at $\tau=\tau_M$, with the radius of the cylinders of symmetry equal to zero there and which is reached by outgoing radial null rays in finite affine time.\\
\\
Case 5. If 
\beas (i)\quad &k^2=2\quad and \quad 0<V_0\le1/4,\\
(ii) \quad &k^2<2\quad and \quad V_0e^{\lambda l_0} = k^2/8,\eeas
then $\tau_M=+\infty$ and there is a curvature singularity on $\mathcal{N}_-$, with the radius of the cylinders of symmetry equal to zero there and which is reached by outgoing radial null rays in finite affine time. 
\end{theorem}

%

\begin{figure}
\includegraphics[width=4.5in]{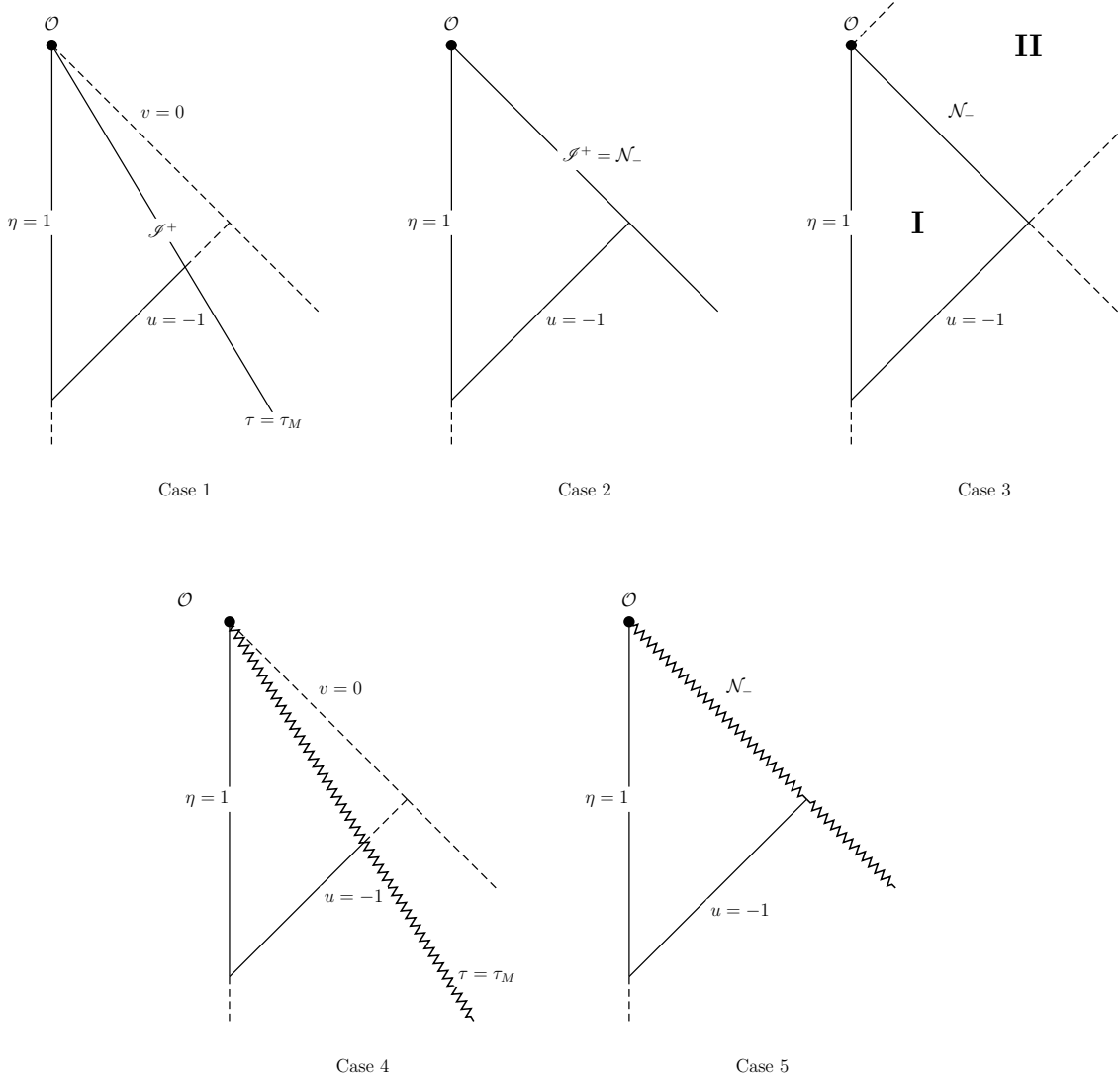}
\caption{The structure of the spacetime for each subcase}
\end{figure}


\section{Existence and uniqueness of solutions with a regular axis}
In this section we present a series of results which culminate in the proof of Theorem 2.1. 
The axis is a singular point of \refb{EFE tau D} and so existence and uniqueness of a solution is not guaranteed. However,
using a fixed point argument, we can prove that for a given inital data set, a unique solution to (50) exists on an interval $[0,\tau_*]$, for some $\tau_*>0$.\\
Note that on the axis, $R$ is zero to first order only, that is, $\dot{R}(0)=\frac{dR}{d\tau}|_{\tau=0}\neq0$. It is convenient to work with a new variable $x=R/\tau$, which is non-zero on the axis, and to look for solutions which are $C^2$. We use a first order reduction and write the system as a set of integral equations according to the following results:
\begin{lemma} Let $x=R/\tau$. Then initial data for $x,\dot{x}$ corresponding to a regular axis are given by
\begin{equation}
x(0) =1,\quad \dot{x}(0) = 0.
\end{equation}
\end{lemma} \begin{proof} Using Taylor's theorem about $\tau=0$ with the assumption that $R\in C^2$, we write $R$ and thus $x$ as 
\begin{subequations}\begin{eqnarray}
R(\tau) =\tau + \ddot{R}(\hat{\tau}(\tau))\frac{\tau^2}{2},\\
x(\tau)=1+ \ddot{R}(\hat{\tau}(\tau))\frac{\tau}{2},
\end{eqnarray}\end{subequations}
for some $\hat{\tau} \in [0,\tau]$. Evaluating $\dot{x}(0)$ from first principles, we find
\begin{equation}
\dot{x}(0) = \lim_{\tau\rightarrow 0}\frac{x(\tau)-x(0)}{\tau} = \lim_{\tau\rightarrow 0}\frac{\ddot{R}(\hat{\tau})\tau}{2\tau} = \lim_{\tau\rightarrow 0}\frac{\ddot{R}(\hat{\tau}(\tau))}{2}.
\end{equation}
Now, $\hat{\tau}\in [0,\tau]$ goes to zero in the limit $\tau\rightarrow 0$ and so $\dot{x}(0)  = \ddot{R}(0)/2 = 0$,
from \refb{EFE tau B}. \end{proof}
\begin{lemma} Let $\boldsymbol{x}=(x_1,x_2,x_3,x_4)$ where $x_1=x,x_2=\dot{x},x_3=l,x_4=\dot{l}.$ Then \refb{EFE tau B},\refb{EFE tau D} and \refb{EFE tau E} are equivalent to the following set of integral equations:
\begin{subequations}\begin{eqnarray}
\label{int x1}x_{1} &= 1 + \int_{0}^{\tau}x_{2}(t)dt,\\
\label{int x2}x_{2} &= \int_{0}^{\tau}\frac{t^2}{\tau^2}x_{1}(t)\alpha(x_{3}(t))dt,\\
\label{int x3}x_{3} &= l_0 + \int_{0}^{\tau}x_{4}(t)dt,\\
\label{int x4}x_{4} &= \int_{0}^{\tau}\frac{tx_{1}(t)}{\tau x_{1}(\tau)}\beta(x_{3}(t))dt,
\end{eqnarray}
where
\begin{equation}
\alpha(x_{3}(t)) = \frac{1}{4}- V_0e^{\lambda x_{3}(t)},\quad
\beta(x_{3}(t))= \frac{1}{4}-\frac{2}{k^2}V_0 e^{\lambda x_{3}(t)}.
\end{equation}
\end{subequations}
\end{lemma} \begin{proof} From \refb{EFE tau B} and $R = \tau x$, we have 
\begin{equation}
\tau\ddot{R} = \tau^2\ddot{x} + 2\tau\dot{x} = \tau^2x\left(\frac{1}{4} - V_0e^{\lambda l}\right)
\end{equation}
which can be integrated to give 
\begin{equation*}
\dot{x}=  \frac{1}{\tau^2}\int^\tau_0t^2x(t)\left(\frac{1}{4} - V_0e^{\lambda l(t)}\right)dt.
\end{equation*}
Equation \refb{EFE tau D} may also be written in the integral form
\begin{eqnarray}
\dot{l}& = \frac{1}{R}\int^\tau_0 R\left(\frac{1}{4} - \frac{2}{k^2}V_0 e^{\lambda l(t)}\right)dt,\\
\nonumber&=\int_{0}^{\tau}\frac{tx_{1}(t)}{\tau x_{1}(\tau)}\beta(x_{3}(t))dt.
\end{eqnarray}
Equations \refb{int x1} and \refb{int x3} follow immediately from the definitions. \end{proof}
A solution of (57) corresponds to a fixed point of the mapping \\$T:\boldsymbol{x} \rightarrow T(\boldsymbol{x})= \boldsymbol{y}=(y_1,y_2,y_3,y_4)$ where
\begin{subequations}\begin{eqnarray}
\label{int y1}y_{1} &= 1 + \int_{0}^{\tau}x_{2}(t)dt,\\
\label{int y2}y_{2} &= \int_{0}^{\tau}\frac{t^2}{\tau^2}x_{1}(t)\alpha(x_{3}(t))dt,\\
\label{int y3}y_{3} &= l_0 + \int_{0}^{\tau}x_{4}(t)dt,\\
\label{int y4}y_{4} &= \int_{0}^{\tau}\frac{tx_{1}(t)}{\tau x_{1}(\tau)}\beta(x_{3}(t))dt.
\end{eqnarray}\end{subequations}
We aim to use Banach's fixed point theorem (the contraction mapping principle)\cite{Banach} to show that $T$ has a unique fixed point. 
We begin by defining the space $\chi$ in which $\boldsymbol{x}$ lies, which we require to be a closed subset of a Banach space. Let $E=C^0([0,\tau_*],\mathbb{R}^4)$, with the norm of a vector $\boldsymbol{x}$ given by
\begin{equation}
||\boldsymbol{x}||_E = \sup_{\tau\in[0,\tau_*]}|\boldsymbol{x}(\tau)| = \sup_{\tau\in[0,\tau_*]}\max_{1\leq i\leq4}|x_i(\tau)|.
\end{equation}
$E$ is therefore a Banach space \cite{Banach}. Let
\begin{eqnarray}
\nonumber\chi(\tau_*,b,B)= \{&\boldsymbol{x}\in C^0([0,\tau_*],\mathbb{R}^4):\boldsymbol{x}(0) =\boldsymbol{x_0},\\
\label{chi}&\sup_{\tau\in[0,\tau_*]}||\boldsymbol{x}-\boldsymbol{x_0}||\leq B, \, \inf_{\tau\in[0,\tau_*]}x_{1}(\tau)\geq b > 0\},
\end{eqnarray}
where $\boldsymbol{x}(0) = (1,0,l_0,0)^T$, and $b<1$. Then $\chi$ is a closed subset of $E$, and is therefore also a Banach space. 
We wish to show that it is possible to choose $\tau_*,b$ and $B$ such that $T$ is a contraction mapping on $\chi$, i.e. $T$ maps $\chi$ into itself
and that there is a number $0<\kappa<1$ such that for any vectors $\boldsymbol{x}^{(1)},\boldsymbol{x}^{(2)} \in \chi$,
\begin{equation}
||\boldsymbol{y}^{(1)}-\boldsymbol{y}^{(2)}||\leq\kappa||\boldsymbol{x}^{(1)}-\boldsymbol{x}^{(2)}||.
\end{equation}
The following four results verify that $T\boldsymbol{x}=\boldsymbol{y}\in\chi$.
\begin{lemma}  The image $T\boldsymbol{x} = \boldsymbol{y}$ of $\boldsymbol{x}\in\chi$, has the same initial data as $\boldsymbol{x}$. That is, $\boldsymbol{y}(0)=\boldsymbol{x}_0$.
\end{lemma} 
\begin{proof} It is straightforward to show that the integral components in (60) equal zero at $\tau=0$ by using the weighted mean value theorem for integrals. 
\end{proof}
\begin{lemma} Let
\begin{equation}
M_1 = \max \left\{B,(B+1)\left(\frac{1}{4}+\delta\right), \frac{(B+1)}{b}\left(\frac{1}{4}+\frac{2\delta}{k^2}\right)\right\}.
\end{equation}
If $\tau_* \leq B/M_1$, then 
\begin{equation}
\label{sup}\sup_{\tau \in [0,\tau_*]}||\boldsymbol{x} - \boldsymbol{x_0}|| \leq  B, \quad \Rightarrow \quad \sup_{\tau \in [0,\tau_*]}||\boldsymbol{y} - \boldsymbol{y}_0|| \leq B,
\end{equation}
\end{lemma} 
\begin{proof} We first note the following inequalities, which hold on $[0,\tau_*]$:
\begin{subequations}\begin{eqnarray}b \leq x_1 \leq B + 1,\quad x_2 \leq B,\quad x_3  \leq  B + |l_0|, \quad x_4 \leq B,\\
|\alpha(x_3)|=\left|\frac{1}{4}- V_0e^{\lambda x_3}\right| \leq\frac{1}{4}+ \delta, \\|\beta(x_3)|=\left|\frac{1}{4}-\frac{2}{k^2}V_0e^{\lambda x_3}\right| \leq \frac{1}{4}+\frac{2\delta}{k^2},
\end{eqnarray}\end{subequations}
where $\delta = |V_0|e^{|\lambda|(B+|l_0|)}$. Now, 
\begin{equation}
\sup_{\tau \in [0,\tau_*]}||\boldsymbol{y} - \boldsymbol{y}_0|| = \sup_{\tau \in [0,\tau_*]}\max(A),
\end{equation}
where 
\bea A = \left\{\left|\int_{0}^{\tau}x_{2}(t)dt\right|,\left|\int_{0}^{\tau}\frac{t^2}{\tau^2}x_{1}(t)\alpha(x_{3}(t))dt\right|,\right.\\
\nonumber\hspace{102pt}\left.\left|\int_{0}^{\tau}x_{4}(t)dt\right|,\left|\int_{0}^{\tau}\frac{tx_{1}(t)}{\tau x_{1}(\tau)}\beta(x_{3}(t))dt\right|\right\}.\eea
We derive a bound for each element of $A$ as follows:
\begin{equation}
\label{int x2}\left|\int_{0}^{\tau}x_{2}(t)dt\right| \leq \int_{0}^{\tau}|x_{2}(t)|dt \leq \int_{0}^{\tau}Bdt=B\tau.
\end{equation}
\begin{eqnarray}
\nonumber\left|\int_{0}^{\tau}\frac{t^2}{\tau^2}x_{1}(t)\alpha(x_{3}(t))dt\right| &\leq  \int_{0}^{\tau}|x_{1}(t)\alpha(x_{3}(t))|dt,\\
&\leq\int_{0}^{\tau}(B+1)\left(\frac{1}{4}+\delta\right)dt,\\
\nonumber&=(B+1)\left(\frac{1}{4}+\delta\right)\tau.
\end{eqnarray}
\begin{equation}
\left|\int_{0}^{\tau}x_{4}(t)dt\right| \leq \int_{0}^{\tau}|x_{4}(t)|dt \leq \int_{0}^{\tau}Bdt = B\tau.
\end{equation}
\begin{eqnarray}
\nonumber\left|\int_{0}^{\tau}\frac{tx_{1}(t)}{\tau x_{1}(\tau)}\beta(x_{3}(t))dt\right| &\leq \int_{0}^{\tau}\left|\frac{x_{1}(t)}{x_{1}(\tau)}\beta(x_{3}(t))\right|dt,\\
&\leq \int_{0}^{\tau}\frac{(B+1)}{b}\left(\frac{1}{4}+\frac{2\delta}{k^2}\right) dt\\
\nonumber&=\frac{B+1}{b}\left(\frac{1}{4}+\frac{2\delta}{k^2}\right)\tau.
\end{eqnarray}
We then have $\max(A)\leq M_1\tau$ and so $\sup_{\tau \in [0,\tau_*]}\max(A)\leq M_1\tau_*$. We choose $\tau_* \leq B/M_1$ and \refb{sup} is satisfied. \end{proof}
\noindent\begin{lemma} If $\tau_*\leq (1- b)/B$, then
\begin{equation}
\label{inf}\inf_{\tau\in[0,\tau_*]}x_1 \geq b> 0 \quad \Rightarrow \quad \inf_{\tau\in[0,\tau_*]}y_1 \geq b > 0.\\
\end{equation}
\end{lemma}
\begin{proof} 
It follows from \refb{int x2} that
that 
\beq
\inf_{\tau\in[0,\tau_*]} \int^\tau_0 x_2(t)dt \geq -B\tau_*.
\eeq
Hence,
\begin{equation}
\inf_{\tau\in[0,\tau_*]} y_1 = 1 + \inf_{\tau\in[0,\tau_*]}\int^\tau_0 x_2(t)dt  \geq 1 - B\tau_*.
\end{equation}
We then choose $\tau_* \leq (1 - b)/B$ so that \refb{inf} is satisfied. Note that $x_1(0)=1 > b$ and so the upper bound on $\tau_*$ is strictly positive. \end{proof}
\begin{proposition} For a given $\boldsymbol{x} \in \chi(\tau_*,B,b)$, with $\tau_*\le\min\{b/M_1,(1-b)/B\}$, we have $T\boldsymbol{x} = \boldsymbol{y}\in\chi$, where $T$ is defined by (37). 
\end{proposition}
\begin{proof} The proof follows immediately from the three preceding lemmas. \end{proof}
\begin{proposition} Let 
\begin{subequations}\begin{eqnarray} \tilde{M}=\max \left\{\frac{1}{4} +\delta,\lambda(B+1)\delta\right\}\\
\bar{M} = \max \left\{\frac{1}{4}+\frac{2\delta}{k^2},\frac{2\lambda}{k^2} (B+1)\delta\right\},\\
M_2 = \max\left\{1,\tilde{M},(B+1)\left(\frac{1}{4}+\frac{2\delta}{k^2}\right)\frac{1}{b^2}+\frac{1}{b}\bar{M}\right\}.
\end{eqnarray}\end{subequations}
The mapping
\beq
T:\chi\rightarrow\chi ;\boldsymbol{x}\mapsto T\boldsymbol{x}=\boldsymbol{y},
\eeq
with $\chi$ defined by \refb{chi}, with $\tau_* < \min\{b/M_1,(1-b)/B,1/M_2\}$, is a contraction mapping, i.e., there exists a number $0<\kappa<1$ such that
\begin{equation}
||T\boldsymbol{x}^{(1)}-T\boldsymbol{x}^{(2)}||_\chi =||\boldsymbol{y}^{(1)}-\boldsymbol{y^{(2)}}||_\chi \leq \kappa||\boldsymbol{x}^{(1)}-\boldsymbol{x}^{(2)}||_\chi,
\end{equation}
for any $\boldsymbol{x}^{(1)},\boldsymbol{x}^{(2)}$ in $\chi$.
\end{proposition}
\begin{proof} 
Recall
\begin{equation}
||\boldsymbol{y^{(1)}}-\boldsymbol{y^{(2)}}||_\chi = \sup_{\tau\in[0,\tau^*]}\max_{1\leq i\leq 4}|y^{(1)}_i-y^{(2)}_i|.
\end{equation}
We show that for each $i$, $1\le i\le4$, we have $|y^{(1)}_i-y^{(2)}_i| \leq a\sigma\tau^*$ where $a$ is some constant and $\sigma = ||\boldsymbol{x}^{(1)}-\boldsymbol{x}^{(2)}||_\chi$.
Then by choosing an appropriate value for $\tau_*$, we show that $T$ is a contraction on the interval $[0,\tau_*]$.
We have
\begin{eqnarray}
\nonumber|y_1^{(1)}-y_1^{(2)}| = \left|\int^\tau_0 x_2^{(1)}(t)-x_2^{(2)}(t)dt\right| &\leq&\int^\tau_0 \left|x_2^{(1)}(t)-x_2^{(2)}(t)\right|dt,\\
& \leq&\int^\tau_0 \sigma dt = \sigma \tau.
\end{eqnarray}
\begin{eqnarray}
\nonumber|y_2^{(1)}-y_2^{(2)}|& =& \left|\int^\tau_0(\alpha^{(1)}(t)x^{(1)}_1(t)-\alpha^{(2)}(t)x^{(2)}_1(t))\frac{t^2}{\tau^2}dt\right|, \\
&\leq& \int^\tau_0\left|(\alpha^{(1)}(t)x^{(1)}_1(t)-\alpha^{(2)}(t)x^{(2)}_1(t))\right|dt,
\end{eqnarray}
where $\alpha^{(j)}= \alpha(x_3^{(j)})$. Let
$p^{(i)}= (x_1^{(i)},x_3^{(i)})^T,
f(p^{(i)}) = \alpha(x_3^{(i)})x_1^{(i)}$ for $i={1,2}$.\\
Then by the mean value theorem, there exists some point $\hat{p}=(\hat{x}_1,\hat{x}_3)$ on the line segment joining $p^{(1)}$ to $p^{(2)}$ such that 
\begin{equation}
\alpha^{(1)}x^{(1)}_1-\alpha^{(2)}x^{(2)}_1 = f(\boldsymbol{p}^{(1)})-f(\boldsymbol{p}^{(2)})=\boldsymbol{\nabla}f(\hat{p})\cdot(p^{(1)}-p^{(2)}).
\end{equation}
Then by the Cauchy-Schwarz inequality we have 
\begin{equation}
\label{C-S}\left|\alpha^{(1)}x^{(1)}_1-\alpha^{(2)}x^{(2)}_1\right|\le\left |\boldsymbol{\nabla}f(\hat{p})\right|\left|p^{(1)}-p^{(2)}\right|.
\end{equation}
Note that $f$ is differentiable everywhere and thus satisfies the hypotheses of the mean value theorem. We have 
\begin{equation}
\boldsymbol{\nabla}f(\hat{p})= \left(\frac{1}{4}-V_0 e^{\lambda \hat{x}_3},-\lambda V_0 e^{\lambda \hat{x}_3}\hat{x}_1\right)^T,
\end{equation}
Using the inequalities $\hat{x}_1\leq B+1$, $\hat{x}_3 \leq B + |l_0|$, we find
\begin{equation}
\label{Grad}\left|\boldsymbol{\nabla}f(\hat{p})\right| \leq \max \left\{\frac{1}{4} +\delta,\lambda(B+1)\delta\right\} = \tilde{M}.
\end{equation}
Using \refb{C-S} and \refb{Grad}  we find $|\alpha^{(1)}x^{(1)}_1-\alpha^{(2)}x^{(2)}_1|\leq \tilde{M}|p^{(1)}-p^{(2)}|$.
Clearly $|p^{(1)}-p^{(2)}|\leq \sigma=||\boldsymbol{x}^{(1)}-\boldsymbol{x}^{(2)}||_\chi$, and so 
\begin{eqnarray}
\nonumber \int^\tau_0\left|y_2^{(1)}-y_2^{(2)}\right|dt &= \int^\tau_0\left|(\alpha^{(1)}(t)x^{(1)}_1(t)-\alpha^{(2)}(t)x^{(2)}_1(t))\right|dt\\
 &\leq \int^\tau_0 \sigma \tilde{M}dt = \sigma \tilde{M} \tau.
\end{eqnarray}
Similarly, we have
\bea\left|y_3^{(1)}-y_3^{(2)}\right| = \left|\int^\tau_0 x_4^{(1)}(t)-x_4^{(2)}(t)dt\right|&\leq\int^\tau_0\left|x_4^{(1)}(t)-x_4^{(2)}(t)\right|dt \\
\nonumber&\leq\int^\tau_0 \sigma dt = \sigma \tau.
\eea
Finally,
\begin{eqnarray} 
\left|y_4^{(1)}-y_4^{(2)}\right| &=\left|\int^\tau_0\left(\frac{\beta^{(1)}x_1^{(1)}(t)}{x_1^{(1)}(\tau)}-\frac{\beta^{(2)}x_1^{(2)}(t)}{x_1^{(2)}(\tau)}\right)\frac{t}{\tau}dt\right|\\
\nonumber &\leq \int^\tau_0\left|\frac{\beta^{(1)}x_1^{(1)}(t)}{x_1^{(1)}(\tau)}-\frac{\beta^{(2)}x_1^{(2)}(t)}{x_1^{(2)}(\tau)}\right|dt\\
\nonumber & = \int^\tau_0\left|\beta^{(1)}x_1^{(1)}(t)\left(\frac{x_1^{(2)}(\tau)-x_1^{(1)}(\tau)}{x_1^{(1)}(\tau)x_1^{(2)}(\tau)}\right) +
\frac{\beta^{(1)}x_1^{(1)}(t)-\beta^{(2)}x_1^{(2)}(t)}{x_1^{(2)}(\tau)}\right|dt\\
\nonumber &\leq \int^\tau_0\left|\beta^{(1)}x_1^{(1)}(t)\right|\frac{\sigma}{b^2}dt+\int^\tau_0\frac{1}{b}\left|\beta^{(1)}x_1^{(1)}(t)-\beta^{(2)}x_1^{(2)}\right|dt = I_1 + I_2,
\end{eqnarray}
using $1/b\geq1/x_1$ and $|x_1^{(2)}-x_1^{(1)}| \leq \sigma$. Using the mean value theorem and the Cauchy-Schwarz inequality again here we find $I_2 \leq \bar{M}\sigma\tau/b$, where 
\beq
\bar{M} = \max \left\{\frac{1}{4}+\frac{2}{k^2}\delta,\frac{2}{k^2}\lambda (B+1)\delta\right\}.\eeq
Using the bounds defined by (60) we find 
\beq
I_1 \leq(B+1)\left(\frac{1}{4}+\frac{2}{k^2}\delta\right)\frac{\sigma}{b^2}\tau.\eeq
So we have 
\beq|y_4^{(1)}-y_4^{(2)}| \leq \left[(B+1)\left(\frac{1}{4}+\frac{2}{k^2}\delta\right)\frac{1}{b^2}+\frac{1}{b}\bar{M}\right]\sigma\tau
\eeq
Gathering these bounds we find that 
$\sup_{\tau\in[0,\tau_*]} \max_{1\leq i\leq 4}|y^1_i-y^2_i| \leq M_2\tau_*\sigma,
$\\
where 
\beq
M_2 = \max\left\{1,\tilde{M},(B+1)\left(\frac{1}{4}+\frac{2}{k^2}\delta\right)\frac{1}{b^2}+\frac{1}{b}\bar{M}\right\}.
\eeq
For $\tau_* < 1/M_2$, $||\boldsymbol{y^{(1)}}-\boldsymbol{y^{(2)}}||_\chi \leq \kappa||\boldsymbol{x^{(1)}}-\boldsymbol{x^{(2)}}||_\chi$ where $0<\kappa<1$. \end{proof}
\begin{proposition} For $\tau_*$ sufficiently small, the mapping $T$ defined above has a unique fixed point on $[0,\tau_*]$.
\end{proposition}
\begin{proof} Given any constants $B $ and $b$, let $m = \min\left\{B/M_1,(1 - b)/B,1/M_2\right\}$. For $\tau_* < m$, then, Propositions 3.1 and 3.2 hold, so $T$ is a contractive mapping from a closed subset $\chi$ of a Banach space $E$, into itself. Using Banach's fixed point theorem completes the proof. \end{proof}
In light of this theorem, we know that there is a unique $\boldsymbol{x}$ on $[0,\tau_*]$, such that $T\boldsymbol{x} = \boldsymbol{x}$. Combining this with (60) shows that there is a unique $\boldsymbol{x}$ such that (57) holds, hence (50) has a unique solution on some interval $[0,\tau_*]$.
\begin{proposition} \refb{EFE tau B} and \refb{EFE tau D} subject to \refb{EFE tau F} have a unique solution on $[0,\tau_*]$, for some $\tau_*>0$.
\end{proposition}
\begin{proof} Proposition 3.3 shows that we have a unique solution for $l$ and $x = R\tau$ on $[0,\tau_*]$, so we have a unique solution for $R$ and $l$. \end{proof}
\noindent\textbf{Proof of Theorem 2.2}\begin{proof} Using Proposition 3.4 we have a unique solution for $R$ and $l$ on some interval $[0,\tau_*]$. Integrating \refb{EFE tau C} we have \begin{equation*}
\phi = \phi_0 +\frac{\tau}{4}+\int^\tau_0 \frac{1-\dot{R}}{R}dt,
\end{equation*}
which gives a unique solution for $\phi$. Equation \refb{EFE tau A} then gives a unique solution for $\gamma$. \end{proof}
\section{Evolution of solutions.}
In this section we determine the global structure of solutions in the region bounded by the axis and $\mathcal{N}_-$, with the exception of a few special cases, which are deferred until the following section. An alternative set of variables proves useful:
\begin{equation}u_1=\dot{R}/R, \qquad u_2= |V_0|e^{\lambda l}, \qquad u_3=\dot{l},\end{equation}
They satisfy
\begin{subequations}\begin{eqnarray}
\label{u1}\dot{u}_1=\frac{1}{4}-\epsilon u_2-u_1^2,\\
\label{u2}\dot{u}_2=\lambda u_2u_3,\\
\label{u3}\dot{u}_3=\frac{1}{4}-\epsilon\frac{2u_2}{k^2}-u_1u_3,\\
u_2(0)=|V_0|e^{\lambda l_0}>0,\qquad u_3(0)=0,
\end{eqnarray}\end{subequations}
where $\epsilon=\mbox{sgn}(V_0)$ and $\lambda=k^2/2-1$. Note that $u_1$ is not defined on $\tau\le0$ and that $\lim_{\tau\rightarrow0^+}u_1=\infty$. Note also that $u_2>0$ by definition.\\
Using results from section 3, there exists $\tau_{M}$ such that $\boldsymbol{u}=(u_1,u_2,u_3)^T$ has a unique solution on $(0,\tau_{M})$. The following standard result proves useful in determining the maximal interval of existence in each case(see, for example, \cite{Tav}).
\begin{theorem} Let $\Psi_{\textbf{a}}(t)$ be the unique solution of the differential equation $\textbf{x}'=\textbf{f}(\textbf{x})$, where $\textbf{f}\in C^1(\mathbb{R}^n)$, which satisfies $\textbf{x}(0)=\textbf{a}$, and let $(t_{\min},t_{\max})$ be the maximal interval of existence on which $\Psi_{\textbf{a}}(t)$ is defined. If $t_{\max}$ is finite, then 
\beq
\lim_{t\rightarrow t^-_{\max}}||\Psi_{\textbf{a}}(t)||=+\infty.
\eeq\end{theorem}
This result may be adapted to our system by defining $(0,\tau_{M})$ as the maximal interval of existence for the unique solution $\boldsymbol{u}(\tau)$ of (94). It follows from Theorem 4.1 that if the components of the solution $u_1,u_2$ and $u_3$ satisfy finite lower and upper bounds for all $\tau\in(0,\tau_{M})$, then we have $\tau_{M}=\infty$. Furthermore, if $\tau_M$ is finite then we have $\lms |u_i|=+\infty$ for at least one $i\in\{1,2,3\}$.\\
The system has three parameters $\{k^2,V_0,l_0$\}. The qualitative picture of solutions depends primarily on the signs of $V_0$ and $\lambda=k^2/2-1$ and so we devote a subsection to each of the four permutations. 
\subsection{$V_0<0,\lambda>0$}
In this case we find that radial null infinity exists along a hypersurface corresponding to a finite value of $\tau$. Note that $\epsilon=-1$ and $k^2>2$ here. 
\noindent\begin{lemma} If $V_0<0$, $\lambda>0$ then $u_3>0,\dot{u}_3>0$ and $u_1>u_3$ for $\tau\in(0,\tau_M)$.\end{lemma}
\begin{proof} Consider
\begin{equation}
\dot{u}_1-\dot{u}_3 = -\epsilon\frac{2\lambda}{k^2}u_2-u_1(u_1-u_3)=\frac{2\lambda}{k^2}u_2-u_1(u_1-u_3).
\end{equation}
for $\epsilon=-1$. Since $u_2>0,\lambda>0,$ it is clear that $u_1-u_3$ cannot cross zero from above, so $u_1>u_3$ for $\tau\in(0,\tau_M)$. Equation \refb{u1} with $\epsilon=-1$ shows that $u_1$ cannot cross $1/2$ from above and so $u_1>1/2$ for all $\tau\in(0,\tau_M)$. Hence, $R>0$ for all $\tau\in(0,\tau_M)$. Equation \refb{EFE tau D} may be integrated to give
\begin{equation}
\label{ldot}\dot{l}=u_3=\frac{1}{R}\int_0^\tau\left(\frac{1}{4}-\epsilon\frac{2u_2}{k^2}\right)R\,d\tau',
\end{equation}
which is clearly positive for $R>0,\epsilon=-1$. Since $u_3(0)=0$ we must have $\dot{u}_3>0$ initially. It is straight-forward to check that at $\dot{u}_3=0$ we have $\ddot{u}_3=(u_1-u_3)u_1u_3$ which is positive for $u_1>u_3>0$. \end{proof}
\begin{lemma} Let $\hat{\lambda}= (\lambda+\sqrt{\lambda^2+16})/4$. If $V_0<0$ and $\lambda>0$, then $u_1>u_2^{1/2}/\hat{\lambda}$ for all $\tau\in(0,\tau_M)$ and there exists $\tau_1\in(0,\tau_M)$ such that $u_1$ is monotonically increasing and bounded above by $\sqrt{1/4+u_2}$ for all $\tau\in(\tau_1,\tau_M)$.\end{lemma}
\begin{proof} First note that $u_1>(1/4+u_2)^{1/2}>u_2^{1/2}/\hat{\lambda}$ on some initial interval, where the second inequality holds due to $\hat{\lambda}>1$. The preceding lemma tells us that $u_3>0,\dot{u}_3>0$, from which it follows that $\dot{u}_2>0,\ddot{u}_2>0$ for all $\tau\in(0,\tau_M)$. Since $u_3<u_1$, the $u_i$ are bounded, and solutions therefore exist, while $u_1$ is decreasing. By inspection of \refb{u1}, with $\epsilon=-1$, $u_1$ is decreasing for $u_1>(1/4+u_2)^{1/2}$. Since $\ddot{u}_2>0$, there must exist some $\tau_1\in(0,\tau_M)$ such that $u_1(\tau_1)=(1/4+u_2(\tau_1))^{1/2}$. Note that at $u_1=(1/4+u_2)^{1/2}$ we have
\beq
\label{u1dot0}\frac{d}{d\tau}\left[{u}_1-\left(\frac{1}{4}+u_2\right)^{1/2}\right]=-\frac{\dot{u}_2}{2(1/4+u_2)^{1/2}}<0.
\eeq
$u_1$ is therefore increasing and bounded above by $(1/4+u_2)^{1/2}$ for all $\tau\in(\tau_1,\tau_{M})$. Now consider
\begin{equation}
\label{sqrt u2}\frac{d}{d\tau}\left(\frac{u_2^{1/2}}{\hat{\lambda}}\right)=\frac{\lambda u_2^{1/2}u_3}{2\hat{\lambda}}<\frac{\lambda u_2^{1/2}u_1}{2\hat{\lambda}}.
\end{equation}
Suppose there exists $\tau_{*}$ such that $u_1(\tau_{*})=u_2^{1/2}(\tau_{*})/\hat{\lambda}$. Then, using \refb{u1} and \refb{sqrt u2}
\begin{eqnarray}
\dot{u}_1(\tau_{*})-\left.\frac{d}{d\tau}\left(\frac{u_2^{1/2}}{\hat{\lambda}}\right)\right|_{\tau=\tau_{*}}&>\frac{1}{4}+u_2(\tau_{*})-\frac{u_2(\tau_{*})}{\hat{\lambda}^2}-\frac{\lambda u_2(\tau_{*})}{2\hat{\lambda}^2}\\
\nonumber&=\frac{1}{4}+\left(1-\frac{1}{\hat{\lambda}^2}-\frac{\lambda}{2\hat{\lambda}}\right)u_2(\tau_{*})=\frac{1}{4}>0,
\end{eqnarray}
so $u_1$ cannot cross $u_2^{1/2}/\hat{\lambda}$ from above. \end{proof}
\begin{lemma} If $V_0<0,\lambda>0$, then there exists $\tau_2\in(0,\tau_M)$ such that $u_1(\tau_2)=k^2u_3(\tau_2)$ and $u_1<k^2u_3$ for all $\tau\in(\tau_2,\tau_M)$.\end{lemma}
\begin{proof} Consider
\begin{equation}
\label{u1-ku3}\dot{u}_1-k^2\dot{u}_3 = -\frac{k^2-1}{4}- u_2-u_1\left(u_1-k^2u_3\right)\le -\frac{1}{4}-u_2,
\end{equation}
provided $u_1>k^2u_3$, which holds initially. If $\tau_M$ is infinite, then the result must follow. If not, then by Theorem 4.1, and since $u_2^{1/2}/\hat{\lambda}<u_1<(1/4+u_2)^{1/2}$ and $u_3<u_1$ for $\tau\in(\tau_1,\tau_M)$, we must have $\lms u_2=\lms u_1=\infty$. We know from Lemma 4.2 that $u_1>u_2^{1/2}/\hat{\lambda}$ for all $0<\tau<\tau_M$, which gives $\dot{u}_1<1/4+\lambda u_2/2$, since $1-1/\hat{\lambda}^2=\lambda/2$. We then have
\beq
\label{int u2}\lim_{\tau\rightarrow\tau_M} \int_{\tau_1}^\tau\frac{\lambda u_2}{2}\,d\tau'>\lim_{\tau\rightarrow\tau_M}\left(u_1-u_1(\tau_1) -\frac{\tau-\tau_1}{4}\right)=\infty.
\eeq
Suppose then that $u_1>k^2u_3$ for all $\tau\in(0,\tau_M)$. Integrating \refb{u1-ku3} and taking the limit gives $\lms (u_1-k^2u_3)=-\infty$, so we have a contradiction. 
\end{proof}
\begin{lemma} If $V_0<0$ and $\lambda>0$, then $\tau_M<\infty$ and for $i=\{1,2,3\}$, 
\beq\lim_{\tau\rightarrow\tau_M} u_i=\infty.\eeq \end{lemma}
\begin{proof}
Using the previous lemma, 
\begin{equation}
\dot{u}_2>\frac{\lambda u_1u_2}{k^2}>\frac{\lambda u_2^{3/2}}{k^2\hat{\lambda}},
\end{equation} 
for $\tau\in(\tau_2,\tau_M)$. Integrating over $[\tau_2,\tau]$ and rearranging we find
\begin{equation}
u_2>\left(\frac{1}{u^{1/2}(\tau_2)}-\frac{\lambda(\tau-\tau_2)}{2k^2\hat{\lambda}}\right)^{-2},
\end{equation}
so we have $\tau_M\le k^2\hat{\lambda}u_2^{-1/2}(\tau_2)/\lambda+\tau_2$ and $\lms u_2=\infty$. It follows directly from Lemma 4.2 that $\lms u_1=\infty$. Using Lemmas 4.1 and 4.3 complete the proof.  \end{proof} 
\noindent\begin{proposition}  For $V_0<0$, $\lambda>0$, the surface corresponding to $\tau=\tau_M$ represents future null infinity and the Ricci scalar decays to zero there.
\end{proposition}
\begin{proof} We aim to show that along outgoing radial null geodesics, an infinite amount of affine parameter time is required to reach the surface $\tau=\tau_M$. These geodesics correspond to the lines $u=u_0$ where $u_0$ is constant. We look for solutions to the equation for null geodesics, which reduces to 
\begin{equation}
\ddot{v}+(2\bar{\gamma}_v+2\bar{\phi}_v)\dot{v}^2 = 0,
\end{equation}
where here the dot denotes differentiation with respect to an affine parameter $\mu$, which is chosen such that $\dot{v}>0$ and $\mu(\tau=0)=0$. Dividing by $\dot{v}$ and integrating, we find
\begin{equation}
e^{2\bar{\gamma}+2\bar{\phi}}\dot{v} = \frac{1}{|u_0|}e^{2\gamma+2\phi}\dot{v}=C,
\end{equation}
with $C>0$.
Substituting $2\gamma+2\phi$ using \refb{EFE tau A} gives
\begin{equation}
\frac{1}{|u_0|}e^{(k^2l+\tau)/2}\dot{v}= C.
\end{equation}
We also have $v=u_0\eta=u_0e^{-\tau}$, and thus $dv=-u_0e^{-\tau}d\tau$, along the geodesics. Integrating then leads to 
\begin{equation}
\label{geo}\frac{1}{|u_0|}\int^v_{v_0} e^{(k^2l+\tau)/2}dv'=\int_0^\tau e^{(k^2l-\tau')/2}d\tau'= C\mu.
\end{equation}
Clearly $e^{k^2l/2}=|V_0|^{-1}e^{l}u_2>u_2$ holds for $\tau$ sufficiently close to $\tau_M$. Using \refb{int u2} then gives $\lim_{\tau\rightarrow\tau_M}\mu=\infty$. 
This confirms that the surface $\tau=\tau_M$ corresponds to radial null infinity.\\
To demonstrate the decay of the Ricci scalar, which we label $\mathcal{R}$, it is convenient to consider the trace of the energy-momentum tensor:
\bea
\nonumber g^{ab}T_{ab}=-2g^{01}\psi_u\psi_v-4V=2|u|e^{-2\gamma-2\phi}\left(-\frac{\eta F'}{u}+\frac{k}{2u}\right)\frac{F'}{u}-\frac{4\bar{V}_0e^{-2F/k}}{|u|}\\
\nonumber=\frac{2e^{-2\gamma-2\phi}}{|u|\eta}\left[\left(-\frac{k\eta l'}{2}+\frac{k}{4}\right)\left(\frac{k\eta l'}{2}+\frac{k}{4}\right)-2V_0e^{\lambda l}\right]\\
\nonumber=\frac{e^{-k^2l/2+\tau/2-c_1}}{|u|}\left(\frac{k^2}{2}\left(\frac{1}{4}-\dot{l}^2\right)-4V_0e^{\lambda l}\right)\\
\label{Ricci}=\frac{e^{\tau/2-c_1}}{|u|}\left(\frac{k^2}{2}\left(\frac{1}{4}-\dot{l}^2\right)e^{-k^2l/2}-4V_0e^{-l}\right)=-\mathcal{R}.
\eea
Note that $1/4-\dot{l}^2<0$ approaching $\tau_M$. Using $\dot{l}=u_3<(1/4+u_2)^{1/2}=(1/4-V_0e^{\lambda l})^{1/2}$ we have $V_0e^{-l}<e^{-k^2l/2}(1/4-\dot{l}^2)<0$. Hence, both terms in the bracket decay to zero as $\tau\rightarrow\tau_M$. Since $\tau_M<\infty$, it follows that $\lim_{\tau\rightarrow\tau_M}\mathcal{R}=0$. \end{proof}
\subsection{$V_0<0,\lambda <0$}
Here we find that the surface $\mathcal{N}_-$ corresponds to a fixed point of the system, is regular and is reached by radial null rays in finite affine time. These are some of the solutions which may be extended into region \textbf{II}.
\begin{lemma} If $V_0<0$ and $\lambda<0$, then $\tau_M=+\infty$ and 
\begin{equation}\lim_{\tau\rightarrow\infty}(u_1,u_2,u_3)=\left(\frac{1}{2},0,\frac{1}{2}\right).
\end{equation}\end{lemma}
\begin{proof} We have seen in the proof of Lemma 4.1 that if $V_0<0$, then $u_3>0$ for all $\tau\in(0,\tau_{M})$. Equation \refb{u2} with $\lambda<0$ then tells us that $u_2$ is monotonically decreasing on $(0,\tau_M)$. Equation \refb{u1dot0} then tells us that $u_1$ cannot cross $(1/4+u_2)^{1/2}$ from above. Since $\dot{u}_1<0$ if $u_1>(1/4+u_2)^{1/2}$, $u_1$ is decreasing and bounded below by this term for all $\tau\in(0,\tau_{M})$. For any $\tau_*\in(0,\tau_{M})$ we then have $u_2<u_2(\tau_*)$ and $1/2<u_1<u_1(\tau_*)$ for all $\tau\in(\tau_*,\tau_{M})$. Hence, for all $\tau\in(\tau_*,\tau_M)$ we have
\beq
\frac{1}{4}-u_1(\tau_*)u_3<\dot{u}_3<\frac{1}{4}+\frac{2u_2(\tau_*)}{k^2}-\frac{u_3}{2},
\eeq
from which it follows that 
\beq
\min\{1/4u_1(\tau_*),u_3(\tau_*)\}<u_3<\max\{1/2+4u_2(\tau_*)/k^2,u_3(\tau_*)\}.
\eeq
We have thus far proven that each of the $u_i$ are bounded above and below for all $\tau\in(\tau_*,\tau_M)$ and so $\tau_M=+\infty$ by Theorem 4.1. Now, it follows from \refb{u1-u3} that $u_1-u_3$ can only change sign once. Recalling $\ddot{u}_3=(u_1-u_3)u_1u_3$ at $\dot{u}_3=0$, $\dot{u}_3$ can only change sign a finite number of times and so $u_3$ must be monotone as $\tau\rightarrow\infty$. Hence, each of the $u_i$ are bounded and monotone in the limit as $\tau\rightarrow\infty$ and so the system must evolve to a fixed point. It is easily checked that the only fixed point of the system \refb{u1}-\refb{u3} consistent with the given analysis is $(1/2,0,1/2)$. 
\end{proof}
We now prove that the metric is regular in the limit as $\tau\rightarrow\infty$ in this case. The following theorem, which may be found in Chapter 9 of \cite{Hartman}, proves useful. 
\begin{theorem} In the differential equation
\beq
\label{ODE}{\boldsymbol{x}}'(t)=E{\boldsymbol{x}}+F({\boldsymbol{x}}),
\eeq
let $F(\boldsymbol{x})$ be of class $C^1$ with $F(0)=0, \partial_{\boldsymbol{x}}F(0)=0$. Let the constant matrix $E$ possess $d>0$ eigenvalues having negative real parts, say, $d_i$ eigenvalues with real parts equal to $\alpha_i$, where $\alpha_1<\ldots<\alpha_r<0$ and $d_1+\ldots+d_r=d,$ whereas the other eigenvalues, if any, have non-positive real parts. If $\alpha_r<\omega<0,$ then \refb{ODE} has solutions ${\boldsymbol{x}=\boldsymbol{x}}(t)$$\neq0$, satisfying
\beq 
\label{H1}||{\boldsymbol{x}}(t)||e^{\omega t} = 0,\qquad \mbox{as}\qquad t\rightarrow+\infty,
\eeq
where $||{\boldsymbol{x}}(t)||$ denotes the Euclidean norm, and any such solution satisfies
\beq 
\label{H2}\lim_{t\rightarrow+\infty}t^{-1}\log||{\boldsymbol{x}}(t)||=\alpha_i,\qquad\mbox{for some }i.
\eeq\end{theorem}
\hbox{}\hfill$\square$\\
%
\begin{proposition} If $V_0<0$ and $\lambda<0$, then the metric is regular in the limit as $\tau\rightarrow\infty$, i.e. on $\mathcal{N}_-$, and outgoing radial null rays reach $\mathcal{N}_-$ in finite parameter time.
\end{proposition}
\begin{proof} 
We define a new system of variables via
\beq
\hat{\boldsymbol{u}}=(\hat{u}_1,\hat{u}_2,\hat{u}_3), \qquad \begin{array}{ll} \hat{u}_1=u_1-\frac{1}{2}\\\hat{u}_2=u_2\\\hat{u}_3=u_3-\frac{1}{2} \end{array}
\eeq
Then it is easy to check that the $\hat{\boldsymbol{u}}$-system of equations is of the form \refb{ODE}, satisfying $F(0)=0,\partial_{\boldsymbol{x}}F(0)=0$, where the matrix
\beq
 E=\left(\begin{array}{ccc}-1&-\epsilon&0\\
0&\lambda/{2}&0\\
-1/{2}&-{2\epsilon}/{k^2}&-{1}/{2}\end{array}\right)
\eeq
has 3 negative eigenvalues, $\lambda/2-,1/2$ and $-1$, of which $\lambda/2$ is the greatest. 
Using \refb{H2} and $\alpha_i\le\lambda/2$, for any $\eps>0$ there exists $T(\eps)$ such that $|\hat{u}_1|\le||\hat{\boldsymbol{u}}||<e^{(\lambda/2+\eps)\tau}$. Note that $\hat{u}_1=\dot{S}/S$, and so
\beq
\label{Sdot/S}-e^{(\lambda/2+\eps)\tau}<\frac{\dot{S}}{S}<e^{(\lambda/2+\eps)\tau}
\eeq
for $\tau>T(\eps)$. Integrating and taking the limit $\tau\rightarrow\infty$ then shows that $0<\lm S<+\infty$. Rearranging \refb{Sphidot} we have
\beq
2\dot{\phi}=\frac{e^{-\tau/2}}{S}-\frac{\dot{S}}{S}, 
\eeq
which may be integrated using \refb{Sdot/S} to show $\lm|\phi|<+\infty$. Hence the metric components $g_{\theta\theta}=|u|e^{2\phi}S^2$ and $g_{zz}=|u|e^{-2\phi}$ are regular on $\cal{N}_-$. Notice, however, that $2\gamma\sim (k^2/4+1/2)\tau$ as $\tau\rightarrow\infty$, by \refb{EFE tau A}, and so the component $g_{uv}=|u|^{-1}e^{2\gamma+2\phi}$ blows up at $\cal{N}_-$. This turns out to be a coordinate singularity and may be avoided by making the transformation $|v|\rightarrow\bar{v}=|v|^{-\lambda/2}$. It is straightforward to check that the metric component is then given by $g_{u\bar{v}}=|u|^{-1}|v|^{1+\lambda/2}e^{2\gamma+2\phi}$. Note that $v=u\eta=ue^{-\tau}$ and $1+\lambda/2=k^2/4+1/2$, so the metric is well behaved in this coordinate system. \\
We also have $e^{(k^2l-\tau)/2}\sim e^{\lambda\tau/2}$ as $\tau\rightarrow+\infty$ and so it follows from \refb{geo} that $\lm\mu<+\infty$.
\end{proof}
\subsection{$V_0>0, \lambda<0$}
There are three subcases here, distinguished the sign of $u_2(0)-k^2/8$. When negative, the solutions have a similar structure to those outlined in the previous section. In the positive case, we have a finite interval of existence and a singularity at $\tau_M$. We deal with the case $u_2(0)=k^2/8$ in section 5.2.
\begin{lemma} If $V_0>0,\lambda<0$ and $u_2(0)<k^2/8$, then $\tau_M=\infty$ and
\begin{equation}\lim_{\tau\rightarrow\infty}(u_1,u_2,u_3)=\left(\frac{1}{2},0,\frac{1}{2}\right).
\end{equation}\end{lemma}
\begin{proof}
Using equation \refb{ldot}, with $\epsilon=1$, and $u_2(0)<k^2/8$, we must have $u_3$ initially positive, since $R$ is initially positive. 
Since $u_3$ cannot cross zero from above while $u_2<k^2/8$ and $u_3>0,\lambda<0$ give $\dot{u}_2<0$,  we have $\dot{u}_2<0, u_2<k^2/8$ and $u_3>0$ for all $\tau\in(0,\tau_M)$. At $\dot{u}_1=0$ we have $\ddot{u}_1=-\dot{u}_2>0$. Given that $\dot{u}_1<0$ initially, it must then hold for all $\tau\in(0,\tau_M)$. Note also that $\dot{u}_1>-\lambda/4-u_1^2$ and so $u_1>\sqrt{-\lambda}/2$ for all $\tau\in(0,\tau_M)$. It then follows that $\dot{u}_3<1/4-\sqrt{-\lambda}u_3/2$, from which it follows that $u_3<1/2\sqrt{-\lambda}$, for all $\tau\in(0,\tau_M)$. Hence, all the $u_i$ are bounded and so $\tau_M=\infty$. The remainder of the proof is analogous to that of Lemma 4.5.
\end{proof}
\begin{proposition} If $V_0<0$ and $\lambda<0$, then the metric is regular in the limit as $\tau\rightarrow\infty$, i.e. on $\mathcal{N}_-$, and outgoing radial null rays reach $\mathcal{N}_-$ in finite parameter time.
\end{proposition}
\begin{proof} Note that the sign of $V_0$ does not affect the arguments in Proposition 4.2 and so the proof is identical.
\end{proof}
\begin{lemma} For $V_0>0,\lambda<0$ and $u_2(0)>k^2/8$, we have $\tau_M<\infty$ and 
\beq\lms u_1=-\infty,\quad\lms u_3=-\infty,\quad \lms u_2=+\infty. \eeq\end{lemma}
\noindent\begin{proof}
In this case we have $u_2>k^2/8,u_3<0$ on some initial interval, using equation \refb{ldot}. Indeed, $u_3$ cannot cross zero from below while $u_2>k^2/8$ and since $u_2$ is increasing for $u_3<0$, these conditions hold for all $\tau\in(0,\tau_{M})$. We then have 
\beq
\dot{u}_1<-\lambda/4-u_1^2,
\eeq
for all $\tau\in(0,\tau_{M})$. For $u_1>\sqrt{|\lambda|}/2$, this may be integrated over $(0,\tau)$ to give
\beq
\label{coth}u_1<m\coth m\tau,
\eeq
where $m=\sqrt{|\lambda|}/2$ and we have used $\lim_{\tau\rightarrow0^+}u_1=\infty$. Note that \refb{coth} automatically holds for $u_1<m$ also, since $\coth m\tau>1$ for all $\tau$.\\
Combining \refb{coth} with \refb{u3} gives 
\beq
\dot{u}_3<-b-(m\coth m\tau)u_3,
\eeq
where $b=2u_2(0)/k^2-1/4>0$. This may be integrated to give
\begin{eqnarray}
u_3<-\frac{b(\cosh m\tau-1)}{m\sinh m\tau},\\
\lambda\int_0^\tau u_3\,d\tau'> -\frac{2\lambda b}{m^2}\log\left(\cosh \left(\frac{m\tau}{2}\right)\right)=8b\log\left(\cosh  \left(\frac{m\tau}{2}\right)\right).
\end{eqnarray}
We then arrive at 
\beq
u_2=u_2(0)\exp\left[\int_0^\tau u_3\,d\tau'\right]>\cosh^{8b}\left(\frac{m\tau}{2}\right).
\eeq
It follows that $\tau_M$ must be finite since, otherwise, we would have $\tau_*\in(0,\tau_M)$ such that $u_2>1/2$ for all $\tau>\tau_*$, which gives $\dot{u}_1<-1/4-u_1^2$ and thus $\lms u_1=-\infty$ for $\tau_M$ finite, using \refb{u1} with $\epsilon=1$. So we must have $\lms |u_i|=\infty$ for at least one of the $u_i$. Now consider $X=u_1-k^2u_3/2$, which satisfies
\beq
\label{X}\dot{X}=-\frac{\lambda}{4}-u_1X>-u_1X. 
\eeq
Note that $X$ is initially positive and at $X=0$ we have $\dot{X}>0$, so $X>0$ for $\tau\in(0,\tau_M)$. Furthermore, if $u_1$ is bounded below for $\tau\in(0,\tau_M)$, then $X$ is bounded above, using \refb{X}, which in turn gives a lower bound for $u_3$. This gives an upper bound on $u_2$, which contradicts $\lms||\boldsymbol{u}(\tau)||=\infty$. Hence, we must have $\lms u_1=-\infty$ and, since $X>0$, $\lms u_3=-\infty$. \\
To complete the proof, we assume $\lms u_2=B+1/4<\infty$ for some constant $B$, and then arrive at a contradiction. Note that $u_2$ is monotone and so the limit must exist. The assumption gives $u_2<B+1/4$, and thus
$\dot{u}_1>-B-u_1^2$, for all $\tau\in(0,\tau_{M})$. Dividing by $u_1$ we have $\dot{u}_1/u_1<-B/u_1-u_1$ for $u_1<0$. Choosing $\tau_0$ such that $u_1(\tau_0)<0$ and integrating over $[\tau_0,\tau]$ gives
\beq
u_1>u_1(\tau_0)\exp\left[\int_{\tau_0}^\tau\frac{-B}{u_1}-u_1\,d\tau'\right].
\eeq
Since $\lms u_1=-\infty$, it must follow that 
\beq
\lms\int_{\tau_0}^\tau u_1\,d\tau'=-\infty.
\eeq
Now, $u_1>k^2u_3/2$ gives $\dot{u}_2>2\lambda u_2u_1/k^2$. Dividing by $u_2$, integrating and using the above, we find $\lms u_2=\infty$, which is our contradiction. 
\end{proof}
\begin{proposition} If $V_0>0$, $\lambda<0$, and $u_2(0)>k^2/8$ then  $\tau_M<\infty$ and there is a singularity at $\tau_M$. 
\end{proposition}
\begin{proof} In terms of $u_2,u_3$, with $V_0>0$, the Ricci scalar is given by
\begin{equation}
\label{Ricci 2}\mathcal{R}=\frac{e^{\tau/2-c_1}}{|u|}\left(\frac{k^2}{2}\left(u_3^2-\frac{1}{4}\right)u_2^{-k^2/2\lambda}+4u_2^{-1/\lambda}\right).
\end{equation}
Using the previous lemma we have $\tau_M<\infty$, $\lms u_2=\infty$ and $\lms u_3=-\infty$. For $\lambda<0$ then, $\lms\mathcal{R}=\infty$. 
\end{proof}
\subsection{$V_0>0, \lambda>0$.}
Similarly to the previous section, we have two different pictures depending on the sign of $u_2(0)-k^2/8$. When positive, $u_3$ is initially negative, and vice-versa. Hence, $u_2$ either starts above $k^2/8$ and is decreasing, or vice-versa. The case $u_2(0)=k^2/8$ is dealt with in the next section. The following results show that in all cases, the maximal interval of existence of solutions is finite and there is a singularity at $\tau_M$.
\begin{lemma} For $V_0>0,\lambda>0$, suppose that $\lms u_3=-\infty$. Then $\lms\mathcal{R}=\infty$.\end{lemma}
\begin{proof} If $\lms u_3=-\infty$ and $\lambda>0$,  then $\lms u_2^{-k^2/2\lambda}\neq0$. Using \refb{Ricci 2} then gives the result. \end{proof}
\begin{lemma}  For $V_0>0$, $\lambda>0$, suppose that $\tau_M<\infty$ and 
$\lms u_1=-\infty,\lms u_3=+\infty.$ Then for any $\tau_{*}<\tau_M$ such that $u_3>0$ for all $\tau\in(\tau_*,\tau_M)$ we have
\beq\lms \int_{\tau_{*}}^\tau u_1 d\tau = -\infty. 
\eeq\end{lemma}
\begin{proof} Note that for $u_3>0$ and $\epsilon=1$ we have $\dot{u}_3/u_3<1/4u_3-u_1$. Integrating then gives
\begin{eqnarray}
u_3<u_3(\tau_{*})\exp\left(\int_{\tau_{*}}^\tau \frac{1}{4u_3}-u_1 d\tau'\right),\\
\lms \exp\left(-\int_{\tau_{*}}^\tau u_1 d\tau'\right)>\lms\frac{u_3}{u_3(\tau_{*})}\exp\left(-\int_{\tau_{*}}^\tau\frac{1}{4u_3} d\tau'\right)=+\infty,
\end{eqnarray}
since $\lms1/4u_3=0$. The result immediately follows. \end{proof}
\begin{lemma} For $V_0>0$, $\lambda>0$, suppose that $\tau_M<\infty$ and\\
$\lms u_1=-\infty,\,\lms u_3=+\infty.$
Then
 \label{Limit}\beq\lms u_2=+\infty,\qquad-\infty<\lms \frac{u_1}{u_3} =L<-k^2-\frac{1}{4}.\eeq\end{lemma}
\begin{proof} First we define $q_1=u_1/u_3$, which satisfies
\beq
\label{q1}\dot{q_1}=\frac{1}{u_3}\left(\frac{1}{4}-u_2\right)-\frac{q_1}{u_3}\left(\frac{1}{4}-\frac{2u_2}{k^2}\right).
\eeq
Now, suppose $u_2$ is bounded above for all $\tau\in(0,\tau_M)$. Then it straightforward to show that $q_1$ must be bounded below, i.e., there exists some $L_*<0$ such that $u_1/u_3>L_*$, for $\tau\in(0,\tau_M)$. We then have $\dot{u}_2/u_2>\lambda u_1/L_*$ for $\tau\in(0,\tau_M)$. Integrating and using Lemma 4.9 then shows that $\lms u_2=+\infty$, and so $u_2$ is unbounded. Given that $u_2$ is monotone increasing for $u_3>0$, which obtains for $\tau$ sufficiently close to $\tau_M$, we must have $\lms u_2=+\infty$.\\ 
It is clear from \refb{q1} that $\dot{q}_1$ is negative if $u_2>k^2/8>1/4$, $u_1<0$ and $u_3>0$, which all hold for $\tau$ sufficiently close to $\tau_M$. Thus, $q_1$ is monotone decreasing for $\tau$ sufficiently close to $\tau_M$ and the limit $\lm q_1=L<0$ exists. We now prove by contradiction that $L$ is finite. Assuming $\lms q_1=-\infty$, there must exist some $\tau_{*}\in(0,\tau_M)$ such that $q_1<-2\lambda$ for $\tau\in(\tau_{*},\tau_M)$. Now define $q_2=u_2/u_1$, which satisfies
\beq
\label{q2}\frac{\dot{q}_2}{q_2}=\lambda u_3-\frac{1}{4u_1}+\frac{u_2}{u_1}+u_1<-\frac{1}{4u_1}+\frac{u_1}{2},
\eeq
for $u_1<0,u_2>1/4,q_1<-2\lambda$. Integrating and taking the limit $\tau\rightarrow\tau_*$ we have
\beq
\label{q3}\lms q_2\ge q_2(\tau_{*}) \lms\exp\left(\frac{1}{2}\int_{\tau_{*}}^\tau u_1-\frac{1}{2u_1} d\tau'\right)=0,
\eeq
where we've used Lemma 4.9, $\lms 1/u_1=0$ and $q_2<0$. It follows that $\lms q_2=0$. Defining $q_3=u_2/u_3$ we have 
\beq
\frac{\dot{q}_3}{q_3}=\lambda u_3-\frac{1}{4u_3}+\frac{2u_2}{k^2u_3}+u_1. 
\eeq
Since $\lms q_2=0$ we can choose $\tau_{*}$ such that we also have 
\beq\frac{2u_2}{k^2u_3}=\frac{2q_2u_1}{k^2u_3}<-\frac{u_1}{4},\eeq
and thus 
\beq\frac{\dot{q}_3}{q_3}<\frac{u_1}{4}-\frac{1}{4u_3},\eeq
for $\tau\in(\tau_{*},\tau_M)$, where we have again used $\lambda u_3<-u_1/2$. Integrating and using Lemma 4.9 then shows that $\lms q_3 = 0$.\\
However, it follows from $u_1<0,u_3>0$ and \refb{q1} that
\beq
\dot{q}_1>-q_3+\frac{2q_3}{k^2}q_1.
\eeq
It is clear that if $\lms q_3=0$, then $\lms q_1$ is finite and so we have a contradiction. Therefore, $L$ must be finite. \\
\\
To estimate $L$, we divide \refb{EFE tau E} across by $u_3^2$ which, in the case $V_0>0$, gives
\beq
\label{constraint}\left(1-\frac{1}{\dot{R}^2}\right)q_1^2+\frac{2u_2}{u_3^2}+k^2q_1+\frac{2+k^2}{8u_3^2}-\frac{k^2}{2}=0,
\eeq
in terms of $R,u_1,u_2,u_3$. Let $q_4=u_2/u_3^2$. To determine the limiting behaviour of $q_4$, we consider its derivative, which may be written as
\beq
\dot{q}_4=\frac{u_2}{u_3}\left(\lambda +\frac{4q_4}{k^2}+2q_1\right)=q_2Y,
\eeq 
where $Y=\lambda +4q_4/k^2+2q_1$. Recall that $q_1$ is monotone decreasing and $u_3>0$ sufficiently close to $\tau_M$, say, on an interval $(\tau_{0},\tau_M)$. Suppose there exists $\tau_1\in(\tau_0,\tau_M)$ such that $\dot{q}_4(\tau_1)=0$. Then we must have $Y(\tau_1)=0$, since $u_2(\tau_1)/u_3(\tau_1)>0$. It is easily shown that this gives $\ddot{q}_4(\tau_1)=q_2(\tau_1)\dot{Y}(\tau_1)$. Moreover, since $\dot{Y}=4\dot{q}_4/k^2+2\dot{q}_1$, we have $\dot{Y}(\tau_1)=2\dot{q}_1(\tau_1)<0$, and thus $\ddot{q}_4(\tau_1)<0$. So $\dot{q}_4$ can only cross zero in $(\tau_0,\tau_M)$ with negative slope, i.e., it may only change sign once on $(\tau_0,\tau_M)$. Therefore, $q_4$ must be monotone close to $\tau_M$, i.e., $\lms q_4$ exists. Suppose $\lms q_4\neq 0$. Then, since $q_4$ is positive, there must exist some $\varepsilon>0, \delta>0$ such that $u_2>\varepsilon u_3^2>1/4$ for $\tau_M-\tau<\delta$. Using this, $u_1/u_3>L$ and \refb{q1} produces
\beq
\dot{q}_1<-\frac{L}{4u_3}+\frac{2\varepsilon u_1}{k^2},
\eeq
for $\tau\in(\tau_M-\delta,\tau_M)$.  It follows that
\beq
\lms q_1< \lms \left(q_1(\tau_M-\delta)+\int_{\tau_M-\delta}^\tau\left(\frac{2\varepsilon u_1}{k^2}-\frac{L}{4u_3}\right)\,d\tau'\right)=-\infty,
\eeq
using Lemma 4.9 and $\lms L/u_3=0$. This contradicts the fact that $L$ is finite and so we have $\lms q_4=0$. 
Taking the limit of \refb{constraint} then yields
\beq
\omega L^2+k^2L-\frac{k^2}{2}=0,
\eeq
where $\omega = \lms (1-\dot{R}^{-2})$. Note that $\omega \le 0$ gives $L\ge 1/2$ which contradicts $L<0$. We then have
\beq
\label{L}L=-\frac{k^2}{2\omega}-\sqrt{\frac{k^4}{4\omega^2}+\frac{k^2}{2\omega}},
\eeq
since the upper root of \refb{L} is positive and, therefore, not allowed. Clearly $\omega < 1$, so $L<-k^2/2-\sqrt{k^4/4+k^2/2}<-k^2-1/4$, if $k^2>1/4$. \end{proof}
\begin{lemma} For $V_0>0$, $\lambda>0$, suppose that $\tau_M<\infty$ and 
\beq
\lms u_1=-\infty,\qquad\lms u_3=+\infty.\eeq
Then $\lms \mathcal{R}=\infty$ and $\lms \mu <+\infty$, i.e. there exists a singularity at $\tau=\tau_M$ which is reached by outgoing null rays in finite affine time.
\end{lemma}\begin{proof} Lemma 4.10 and $\lambda>0$ give $\lms u_2=\lms l =\infty$. Using equation \refb{Ricci}  we have 
\beq
\lms\mathcal{R}=\frac{k^2e^{\tau_M/2-c_1}}{2|u|}\lms \left(e^{-k^2l/4}u_3\right)^2.
\eeq
Define $Z=e^{-k^2l}u_1$, which satisfies
\beq
\dot{Z}=e^{-k^2l/4}\left(\frac{1}{4}-u_2-u_1^2-k^2u_1u_3\right)<\left(-u_1-k^2u_3\right)Z,
\eeq
for $u_2>1/4$. Using Lemma 4.10, we may choose some $\tau_{*}\in(0,\tau_M)$ such that $u_1/u_3<-k^2-1/8$ and $Z<0$ for all $\tau\in(\tau_{*},\tau_M)$. We then have $\dot{Z}/Z>u_3/8>-u_1/8L$
for all $\tau\in(\tau_{*},\tau_M)$. Integrating and using Lemma 4.9 then proves $\lms Z =-\infty$.\\
Hence, $\lms e^{-k^2l/4}u_3=L^{-1}\lms Z= +\infty$, which gives $\lms\mathcal{R}=+\infty$. 
\end{proof}
\begin{lemma}Suppose there exists $\tau_{0}\in(0,\tau_M)$ such that $u_1(\tau_{0})\le-1/2$. Then $\tau_M<\infty$ and
$\lim_{\tau_\rightarrow\tau_{M}^-}u_1=-\infty.$\end{lemma}
\begin{proof} We define a new variable $\bar{u}_1=u_1+1/2$, which satisfies
\begin{eqnarray}
\dot{\bar{u}}_1=\bu1-u_2-\bu1^2 <\bu1-\bu1^2.
\end{eqnarray}
It is clear that if $\bu1 (\tau_{0})\le0$ then there exists some $\tau_1>\tau_{0}$ such that $\lim_{\tau_\rightarrow\tau_{1}^-}\bu1=-\infty.$ Then we must have $\tau_M\le \tau_1$ and, using Theorem 4.1, $\lms |u_i|=+\infty$ for some $i$. Suppose that $\lms u_1>-\infty$. It is clear from \refb{u2} that $u_2$ is finite provided $u_3$ and $\tau$ are finite and so we must have $\lms |u_3|=\infty$. If $\lms u_3=+\infty$, it follows from \refb{ldot} and the fact that $u_2>0$ and $0<R<R(\tau_0)$ that $\lms R=0$. Note that $R<R(\tau_0)$ follows from $u_1<0$ here. Note also that 
\beq
R=R(\tau_0)\exp\left(\int^\tau_{\tau_0} u_1\,d\tau'\right),\eeq
and so $\lms R=0$ implies that
\beq
\lms \int^\tau_{\tau_0} u_1\,d\tau'=-\infty,
\eeq
from which it must follow that $\lms u_1=-\infty$. If $\lms u_3=-\infty$ then we have either $\lms R=0$ or 
\beq
\label{u2 int}\lms \int_0^\tau u_2\,d\tau'=+\infty,
\eeq
where we have used \refb{ldot} and fact that $R$ is bounded above again. It follows immediately from \refb{u1} and \refb{u2 int} that $\lms u_1=-\infty$ in this case also. Hence $\tau_1=\tau_M$ and the proof is complete.  
\end{proof}
\begin{lemma} If $V_0>0,\lambda>0$ and $u_2(0)>k^2/8$, then there exists $\tau_0\in(0,\tau_M)$ such that $u_1(\tau_0)=0$ and $u_2>k^2/8$ for all $\tau\in[0,\tau_0]$. 
\end{lemma}
\begin{proof} Recall that $u_2(0)>k^2/8$ gives $u_3<0,\dot{u}_3<0$ on some initial interval. Differentiating \refb{u3} gives
\begin{equation}
\label{u3DD}\ddot{u}_3=\left(\frac{2u_2}{k^2}-\frac{1}{4}+u_1^2\right)u_3-u_1\dot{u}_3.
\end{equation}
At $\dot{u}_3=0$ we have $\ddot{u}_3=(u_1-u_3)u_1u_3$, which is negative for $u_1>0>u_3$, and so $\dot{u}_3<0$ holds while $u_1>0$. We then have $0<u_2<u_2(0)$ and $\dot{u}_3>1/4-2u_2(0)/k^2$ for $u_1>0$. Hence, the $u_i$ are all bounded above and below for $u_1>0$, and so either $\tau_M=\infty$, or there exists $\tau_0$ such that $u_1(\tau_0)=0$. Consider \refb{X} with $u_1>0$ and $u_3<0$, which give $\dot{X}<-\lambda/4-u_1^2$. It is obvious that $X$ or $u_1$ must cross zero in finite $\tau$. However, $u_1<X$ if $u_3<0$ and so there must exist $\tau_0$ such that $u_1(\tau_0)=0$. We then have $\dot{u}_3(\tau_0)=1/4-2u_2(\tau_0)/k^2<0$,
from which $u_2(\tau_0)>k^2/8$ immediately follows. The fact that $\dot{u}_2=\lambda u_2u_3<0$ for $\tau\in(0,\tau_0]$ completes the proof. \end{proof} 
\begin{lemma} If $V_0>0,\lambda>0,u_2(0)>k^2/8$, then $u_1<0,u_3<0,\dot{u}_3<0,\ddot{u}_3<0$ for all $\tau\in(\tau_0,\tau_M)$, where $u_1(\tau_0)=0$.\end{lemma}
\begin{proof} If $u_2\ge k^2/8$ then $\dot{u}_1\le-\lambda/4-u_1^2<0$. If $u_2<k^2/8$ and $\dot{u}_3<0$, then by \refb{u3} we have $u_1u_3>0$. For $\tau\in(\tau_0,\tau_M)$ then, $u_1<0$ holds while $u_3<0,\dot{u}_3<0$ hold.  
Using the preceding lemma and \refb{EFE tau B}, we have $\ddot{R}<-\lambda R/4$ for all $\tau\in[0,\tau_0]$. This may be integrated to give $R<m^{-1}\sin m\tau\le m^{-1}$, which gives $R^{-2}>m^2=\lambda/4$, for $\tau\in[0,\tau_0].$ Now, \refb{EFE tau E} may be rearranged to give
\beq
\label{constraint2}\frac{1}{R^2}-\frac{\lambda}{4}+\lambda u_1^2+\frac{k^2}{2}\left(u_1-u_3\right)^2=k^2\left(\frac{2u_2}{k^2}-\frac{1}{4}+u_1^2\right). 
\eeq
Then, for $R^{-2}>\lambda/4,\lambda>0$ we must have 
\beq
\frac{2u_2}{k^2}-\frac{1}{4}+u_1^2>0.
\eeq
We see from equation \refb{u3DD} that this inequality, along with $u_1<0,u_3<0,\dot{u}_3<0$, gives $\ddot{u}_3<0$, which preserves $\dot{u}_3<0$. Since $R^{-2}>\lambda/4$ holds for $u_1=\dot{R}/R<0$,
we must have $u_1,u_3,\dot{u}_3,\ddot{u}_3$ negative for $\tau \in[\tau_0,\tau_M)$. \end{proof}
\begin{lemma}  If $V_0>0,\lambda>0$ and $u_2(0)>k^2/8$, then $\tau_M<\infty$ and $\lms u_1=\lms u_3 =-\infty$.\end{lemma}
\begin{proof} We know that there exists $\tau_0$ such that $\dot{u}_3<0,\ddot{u}_3<0$ for all $\tau\in(\tau_0,\tau_M)$. Supposing that $\tau_M=+\infty$, then we must have $\lms u_3=-\infty$. We must also have $u_1>-1/2$ for all $\tau\in(0,\tau_M)$, by Lemma 4.12. 
We also know from the preceding proof that $R^{-2}>\lambda/4$ for all $\tau\in(0,\tau_M)$. Equation \refb{constraint2} then gives
\beq
\frac{k^2}{2}(u_1-u_3)^2<k^2\left(\frac{2u_2}{k^2}-\frac{1}{4}+u_1^2\right).
\eeq
If $u_1>-1/2$, then the lefthand side blows up at $\tau_M$ which, given that $u_2$ is finite, is a clear contradiction. Hence, $u_1$ crosses $-1/2$ at some finite $\tau$ and $\lms u_1=-\infty$ for $\tau_M$ finite. Dividing \refb{u1} by $u_1$, integrating and taking the limit $\tau\rightarrow\tau_M$ we find
\beq
\lms\int_{\tau_{*}}^\tau\frac{1-4u_2}{4u_1}-u_1\,d\tau'=\lms\log\frac{u_1}{u_1(\tau_{*})}=\infty,
\eeq
where $\tau_{*}$ is chosen such that $u_1<0$ for $\tau\in[\tau_{*},\tau_M)$. Since $\lms(1-4u_2)/4u_1=0$, it follows that 
\beq
\lms\int_{\tau_{*}}^\tau u_1\,d\tau'=-\infty.
\eeq
Integrating $\dot{u}_3/u_3$ and taking the limit we find
\beq
\lms\log\frac{u_3}{u_3(\tau_{*})}=\lms\int_{\tau_{*}}^\tau\frac{k^2-8u_2}{4k^2u_3}-u_1\,d\tau'=+\infty,
\eeq
using the fact that $(k^2-8u_2)/4k^2u_3$ is bounded for $\tau\in(0,\tau_M)$. The result immediately follows. \end{proof}
Note that it follows from this result and Lemma 4.8 that there is a singularity at $\tau=\tau_M$ .
\begin{lemma} If $V_0>0,\lambda>0$ and $u_2(0)<k^2/8$, then $\tau_M>\pi/2m$ and $u_1(\pi/2m)>0,u_2(\pi/2m)<k^2/8$, $R(\pi/2m)>m^{-1}$ and $k^2u_3/2>u_1$ for all $\tau\in[\pi/2m,\tau_M)$.\end{lemma}
\begin{proof} For $u_2<k^2/8$ we have $\dot{u}_1>-\lambda/4-u_1^2$. Integrating over $(0,\tau)$ gives 
\begin{equation}
\label{u1 LB}u_1>m\cot (m\tau),
\end{equation}
where we have used $\lim_{\tau\rightarrow0^+}u_1=\infty$. While $u_3>0$ we have $u_2>u_2(0)$ which, combined with the above, gives
\begin{equation}
\label{u3 UB}\dot{u}_3<\frac{1}{4}-\frac{2u_2(0)}{k^2}-(m\cot m\tau)u_3.
\end{equation}
Integrating over $(0,\tau)$ gives
\beq
u_3<\frac{b(1-\cos m\tau)}{m\sin m\tau}=\frac{b\sin(m\tau/2)}{m\cos(m\tau/2)},
\eeq
where $b=1/4-2u_2(0)/k^2$. Integrating again we find 
\bea
\lambda\int_{0}^{\tau} u_3\,d\tau'<-\frac{2\lambda b}{m^2}\log\left[\cos\left(\frac{m\tau}{2}\right)\right]=-8b\log\left[\cos\left(\frac{m\tau}{2}\right)\right],
\eea
and so using $\dot{u}_2=\lambda u_2u_3$, 
\bea
\label{u2 UB}u_2<u_2(0)\cos^{-8b}\left(\frac{m\tau}{2}\right). 
\eea
Note that $u_3$ cannot cross zero from above if $u_2<k^2/8$. The bounds $u_3>0$, \refb{u1 LB},\refb{u3 UB} and \refb{u2 UB} therefore hold, and solutions exist, as long as $u_2<k^2/8$ holds. Assuming $\tau_M>\pi/2m$ we have $u_2(\pi/2m)<2^{4b}u_2(0)$. Letting $z=8u_2(0)/k^2<1,$ and using $4b=1-8u_2(0)/k^2=1-z$, we have
\begin{equation}
\frac{8u_2(\pi/2m)}{k^2}<2^{1-z}z\le1,
\end{equation}
for all $z\le1$, which is equivalent to $u_2(\pi/2m)<k^2/8$. Our assumption is then validated. So we have $u_1(\pi/2m)>0$ from \refb{u1 LB}, and it is straightforward to show that $R>m^{-1}\sin m\tau$ on $[0,\pi/2m]$, which gives $R(\pi/2m)>m^{-1}$. \\
Recall $X=u_1-k^2u_3/2$, which satisfies
\begin{equation}
\dot{X}<-\frac{\lambda}{4}-X^2,
\end{equation}
provided $u_3>0,X\ge0$, using \refb{X}. Integrating over $(0,\tau)$ we find $X<m\cot m\tau$. Since $\cot m\tau=0$ at $\tau=\pi/2m$ and $u_3>0$ for $\tau\in(0,\pi/2m)$, there must exist $\tau_*\in(0,\pi/2m)$ such that $X(\tau_*)=0$. Note also that $X$ cannot cross zero from below if $\lambda>0$ and so $X<0$ for $\tau\in(\tau_*,\tau_M)$. \end{proof}
\begin{lemma} If $V_0>0,\lambda>0$ and $u_2(0)<k^2/8$, then there exists $\tau_0\in(0,\tau_M)$ such that $u_1(\tau_0)=0, u_3(\tau_0)>0$ and $R(\tau_0)>m^{-1}.$
\end{lemma}
\begin{proof} Suppose $u_1>0$ for all $\tau\in(0,\tau_M)$. Then $u_3>0$ for all $\tau\in(0,\tau_M)$, using the previous lemma. We then have $\dot{u}_3<1/4$ for all $\tau\in(0,\tau_M)$, which gives a finite upper bound on $u_3$, and thus $u_2$, for finite $\tau$. Hence, $\tau_M=+\infty$. If $u_2\le1/4$ we have 
\beq
\dot{u}_3\ge\frac{\lambda}{2k^2}-u_1u_3>\frac{\lambda}{2k^2}-u_1(\tau_*)u_3,\eeq
for any $\tau_*\in(0,\tau_M)$. It follows that $u_3>u_{m}=\min\{u_3(\tau_*),\lambda/2k^2u_1(\tau_*)\}$ for $u_2\le1/4$ and $\tau\in(\tau_*,\tau_M)$. This gives $\dot{u}_2>\lambda u_m u_2$, and so there must exist $\tau_{**}\in(\tau_*,\tau_M)$ such that $u_2>1/4$ for all $\tau\in(\tau_{**},\tau_M)$. By inspection of \refb{u1}, there must then exist $\tau_0\in(\tau_{**},\tau_M)$ such that $u_1(\tau_0)=0$. Using Lemma 4.16, we have $\tau_0>\pi/2m$ and thus $u_3(\tau_0)>0$ and $R(\tau_0)>R(\pi/2m)>m^{-1}$. \end{proof}
\begin{lemma} For $V_0>0,\lambda>0$, suppose there exists $\tau_0\in(0,\tau_M)$ such that $u_1(\tau_0)=0$ and $u_2(\tau_0)>k^2/8$. 
Then $\tau_M<+\infty$ and $\lms u_1=-\infty,\lms u_3 = +\infty$. \end{lemma}
\begin{proof}
Equation \refb{constraint} may also be written as 
\beq
\label{ku3}u_1^2-\frac{1}{R^2}+\frac{\lambda}{4}-\frac{k^2u_3^2}{2}=k^2\left(\frac{1}{4}-\frac{2u_2}{k^2}-u_1u_3\right)=k^2\dot{u}_3. 
\eeq
Now define 
\beq
\Gamma =u_1^2-\frac{1}{R^2}+\frac{\lambda}{4},
\eeq
which satisfies
\bea
\nonumber\dot{\Gamma}=2u_1\dot{u}_1+\frac{2\dot{R}}{R^3}&=2u_1\left(\frac{1}{4}-u_2-u_1^2+\frac{1}{R^2}\right)\\
\label{Gamma dot}&=2u_1\left(\frac{k^2}{8}-u_2-\Gamma\right).
\eea
If $\Gamma>0,u_2>k^2/8$ and $u_1<0$, then $\dot{\Gamma}>0$. From the hypothesis we have $\Gamma(\tau_0)>0$ and as long as $u_2>k^2/8$ holds we have $\dot{u}_1<-\lambda/4-u_1^2$. Equation \refb{ku3} tells us that $\dot{u}_3>0$ if $u_3<\sqrt{2\Gamma/k^2}$ and so we have $u_3>0$, which gives $u_2>k^2/8$, while $\Gamma>0$. Hence, $\Gamma>0, u_2>k^2/8, u_1<0$ and $u_3>0$ hold for all $\tau\in(\tau_0,\tau_M)$. We then have $\dot{u}_1<-\lambda/4-u_1^2$ for $\tau\in(\tau_0,\tau_M)$, $\tau_M<+\infty$ and $\lms u_1 = \lms X=-\infty$. Integrating $\dot{X}/X=-\lambda/4X-u_1$, then shows that 
\beq
\lms\int_{\tau_*}^\tau u_1 = -\infty,
\eeq
where $\tau_*$ is chosen such that $u_1(\tau_*)<0$. 
We can use this to show $\lms \Gamma=+\infty$ by integrating \refb{Gamma dot}. Since $\dot{u}_3<0$ for $u_3>\sqrt{2\Gamma/k^2}$, we must have $u_3<\sqrt{2\Gamma/k^2}$ and $\dot{u}_3>0$, for $\tau$ sufficiently close to $\tau_M$ and so $\lms u_3$ must exist. Now suppose $\lms u_3<+\infty$. We then have $\lms u_2<+\infty$ and 
\beq
\lms\log\left(\frac{u_3}{u_3(\tau_*)}\right)=\lms\int_{\tau_*}^\tau \left(\frac{1}{4u_3}-\frac{2u_2}{k^2u_3}-u_1\right)\,d\tau'=+\infty,
\eeq
since $1/4u_3-2u_2/k^2u_3$ is bounded above and below under the assumption. Hence, we have $\lms u_3=+\infty$, by contradiction. \end{proof}
\begin{lemma}  For $V_0>0,\lambda>0$, suppose there exists $\tau_0\in(0,\tau_M)$ such that $u_1(\tau_0)=0,u_2(\tau_0)<k^2/8$. Then $\tau_M<+\infty$ and $\lms u_1=-\infty$, $\lms u_3=+\infty$.
\end{lemma}
\begin{proof} Lemma 4.17 tells us that $u_3(\tau_0)>0$. It is clear from \refb{u3} that if $u_1<0,u_2<k^2/8$ and $u_3>0$, then $\dot{u}_3>0$, which together give $\dot{u}_2>0,\ddot{u}_2>0$. While $|u_1|$ and $u_2$ remain bounded above, $u_3$ remains bounded above also, so either there exists $\tau_*\in(\tau_0,\tau_M)$ such that $u_2(\tau_*)=k^2/8, u_1(\tau_*)<0$, or we have $\lms u_1=-\infty$ for $\tau_M<+\infty$. Suppose the former is true. Then we have $\dot{u}_3(\tau_*)>0$. It follows from \refb{u3DD} that $\dot{u}_3$ cannot cross zero from above if $u_2\ge k^2/8$ and $u_3>0$. Hence, $u_2>k^2/8, u_3>0,\dot{u}_3>0$ hold for all $\tau\in(\tau_*,\tau_M)$. We then have $\dot{u}_1<-\lambda/4-u_1^2$, from which it follows that $\lms u_1 = -\infty$ for $\tau_M<+\infty$ in this case also. Given that, in both cases, $\dot{u}_3>0$ for all $\tau\in[0,\tau_M)$, then $\lms u_3$ must exist. A similar argument to one given in the preceding lemma gives $\lms u_3=+\infty$. Using Lemma 4.10 then gives $\lms u_2=+\infty$ and so such a $\tau_*$ does exist after all. \end{proof}
\begin{lemma} For $V_0>0,\lambda>0$, suppose there exists $\tau_0\in(0,\tau_M)$ such that $u_1(\tau_0)=0,u_2(\tau_0)=k^2/8$. Then $\tau_M<+\infty$ and $\lms u_1=-\infty$, $\lms u_3=+\infty$.
\end{lemma}
\begin{proof}At $\tau_0$ we have $\dot{u}_3=\ddot{u}_3=0$, and it is not hard to check that at the third derivative of $u_3$ reduces to $-\dot{u}_1(\tau_0)u_3(\tau_0)^2>0$. A similar argument to the one given above then shows that $u_2>k^2/8,\dot{u}_3>0$ obtain for $\tau\in(\tau_0,\tau_M)$ and the rest follows in a similar fashion. \end{proof}
\begin{proposition} For $V_0>0,\lambda>0, u_2(0)\neq k^2/8$ we have $\tau_M<\infty$ and there is a curvature singularity at $\tau_M$, which is reached by outgoing null rays in finite affine time. \end{proposition}
\begin{proof} For $u_2(0)>k^2/8$, Lemma 4.15 tells us that $\lms u_1=\lms u_3=-\infty$. Lemma 4.8 then tells us that $\lms \mathcal{R}=\infty$ in this case. Clearly $u_2$ is bounded above for all $\tau\in[0,\tau_M)$ in this case and by inspection of \refb{geo} we see that $\lms \mu<+\infty$.  
In the case $u_2(0)<k^2/8$, Lemma 4.17  shows that there exists $\tau_0\in(0,\tau_M)$ such that $u_1(\tau_0)=0$. Depending on the sign on $u_2(\tau_0)-k^2/8$, one of the three preceding lemmas shows that $\lms u_1=-\infty$ and $\lms=+\infty$. Lemma 4.11 then gives $\lms\mathcal{R}=+\infty.$ To show that $\lms \mu<+\infty$ in this case, we recall from Lemma 4.11 that $e^{-k^2l}u_1=Z<Z(\tau_*)$ for some $\tau_*\in(0,\tau_M)$. This gives $e^{k^2l}<u_1/Z(\tau_*)$, which in turn gives $e^{k^2l/2}<(u_1/Z(\tau_*))^{1/2}$. Now let $p=(-u_1)^{1/2}$ and consider
\beq 
\dot{p}=\frac{1}{2p}\left(-\frac{1}{4}+u_2+u_1^2\right)>\frac{p^3}{2},\eeq
for $u_2>1/4$. Dividing by $p^2$ and integrating we have 
\beq
\label{int p}\int_{\tau_*}^{\tau}\frac{\dot{p}}{p^2}\,d\tau'=\frac{1}{p(\tau_*)}-\frac{1}{p}>\int_{\tau_*}^\tau \frac{p}{2}\,d\tau'. \eeq
Using equation \refb{geo} we have 
\bea
C\mu = \int_0^\tau e^{k^2l/2-\tau'/2}d\tau'&<\int_0^{\tau_*} e^{k^2l/2}d\tau'+\int_{\tau_*}^\tau \left(\frac{u_1}{Z(\tau_*)}\right)^{1/2}d\tau'\\
\nonumber&=\int_0^{\tau_*} e^{k^2l/2}d\tau'+\int_{\tau_*}^\tau \frac{p}{(-Z(\tau_*))^{1/2}}d\tau'.\eea
Taking the limit and using \refb{int p} we find that $\lms \mu<+\infty$. \end{proof}
\section{Exact solutions}
\subsection{$k^2=2$} In this case we have $\lambda =0$ which gives us constant potential $V=V_0$. We then have
\begin{equation}
\label{RDD exact}\ddot{R}=\frac{d}{d\tau}R\dot{l}=\left(\frac{1}{4}-V_0\right)R.
\end{equation}
\begin{proposition} If $k^2=2$ and  $0<V_0\le1/4$, then there is a curvature singularity along $\mathcal{N}_-$ which is reached in finite affine time.\end{proposition}
\begin{proof}
In the case $V_0<1/4$, solutions of \refb{RDD exact} in terms of $S$ are given by
\beq
\label{soln} S=\upsilon^{-1}e^{-\tau/2}\sinh\upsilon\tau,\qquad l=l_0+\log\left[\frac{1}{2}(1+\cosh\upsilon\tau)\right],
\eeq	
where $\upsilon=\sqrt{1/4-V_0}$. We also have 
\beq
\lm\dot{l} =\lm\frac{\upsilon\sinh\upsilon\tau}{1+\cosh\upsilon\tau}=\upsilon.
\eeq
For $k^2=2$ we then have
\beq
\label{Ric soln}\lim_{\tau\rightarrow\infty}\mathcal{R}=\lm \frac{e^{\tau/2-l}}{|u|}\left(\frac{1}{4}-\dot{l}^2+4V_0\right)=\lm \frac{10V_0e^{\tau/2-l_0}}{|u|(1+\cosh\upsilon\tau)}.
\eeq
The solution to the geodesic equation \refb{geo} reduces to
\begin{equation}
\label{geo soln}\frac{1}{2}\int_0^\tau e^{l_0-\tau'/2}(1+\cosh\upsilon\tau')d\tau'=C\mu.
\end{equation}
Note that $V_0>0$ gives $\upsilon<1/2$ for which
\beq
\lm S = 0,\qquad \lm\mathcal{R}=\infty,\qquad \lm \mu<\infty. 
\eeq
For $V_0=1/4$ we have
\beq
S=\tau e^{-\tau/2},\qquad l=l_0, \qquad \mathcal{R}=\frac{5e^{\tau/2-l_0}}{4|u|},
\eeq
which give
\beq
\lm S=0,\qquad \lm\mathcal{R}=\infty,\qquad\lm \mu <\infty. 
\eeq\end{proof}
\begin{proposition} If $k^2=2$, $V_0<0$, then $\mathcal{N}_-$ corrseponds to radial null infinity and the Ricci scalar decays to zero there.\end{proposition}
\begin{proof} In the case $\upsilon>1/2\,(V_0<0)$, \refb{soln},\refb{Ric soln} and \refb{geo soln} tell us that 
\beq
\lm S = \infty,\qquad \lm\mathcal{R}=0, \qquad\lm \mu = \infty.
\eeq
We remind the reader that we are not considering the case $V_0=0 (\upsilon = 1/2)$. \end{proof}

\begin{proposition} If $k^2=2$, $V_0>1/4$, there exists a curvature singluarity along $\tau=\pi/\bar{\upsilon}$ where
$\bar{\upsilon}=\sqrt{V_0-1/4}.$\end{proposition}
\begin{proof}
In the case $V_0>1/4$, solutions to \refb{RDD exact} are given by
\beq S=\bar{\upsilon}^{-1} e^{-\tau/2}\sin\bar{\upsilon}\tau ,\qquad l=l_0+\log\left[\frac{1}{2}(1+\cos\bar{\upsilon}\tau)\right],
\eeq
where $\bar{\upsilon}=\sqrt{V_0-1/4}$. At $\bar{\upsilon}\tau=\pi$ we have $S=0$ and 
\beq\lim_{\tau\rightarrow\pi/\bar{\upsilon}}l=-\infty,\qquad\lim_{\tau\rightarrow\pi/\bar{\upsilon}}\dot{l}=-\lim_{\tau\rightarrow\pi/\bar{\upsilon}}\bar{\upsilon}\tan\left(\frac{\bar{\upsilon}\tau}{2}\right)=-\infty,\eeq
which give
$\lim_{\tau\rightarrow\pi/\bar{\upsilon}}\mathcal{R}=\infty.$ \end{proof}
\subsection{$V_0e^{\lambda l_0}=k^2/8$}
\begin{lemma} $u_3$ is monotone in a neighbourhood of the axis. \end{lemma}
\begin{proof} Note that there exists $\tau_1\in(0,\tau_M)$ such that $u_1>u_3$ and $u_1>0$ hold for $\tau\in(0,\tau_1)$. Suppose there exists $\tau_0\in(0,\tau_1)$ with $\dot{u}_3(\tau_0)=0$. We then have $\ddot{u}_3(\tau_0)=(u_1(\tau_0)-u_3(\tau_0))u_1(\tau_0)u_3(\tau_0)$, which has the same sign as $u_3(\tau_0)$, since $u_1(\tau_0)>u_3(\tau_0), u_1(\tau_0)>0$. So either $u_3(\tau_0)<0$ and is a local max, or $u_3(\tau_0)>0$ and is a local min. Since $u_3(0)=0$, in the former case we must then have $\tau_*\in(0,\tau_0)$ such that $u_3(\tau_*)<0$ is a local min, which is contradiction. Similarly for the latter case. Hence, $u_3$ is monotone on $(0,\tau_1)$. \end{proof}
\begin{lemma} If $u_2(0)=k^2/8$,  then $u_2=k^2/8$ and $u_3=0$ for all $\tau\in[0,\tau_M)$. \end{lemma}
\begin{proof} First note that $u_2=k^2/8,u_3=0$ is an invariant manifold of the system (97) with $\epsilon=-1$. The system (97) is not defined at $\tau=0$ and so we must show that there exists $\tau_0>0$ such that $u_2(\tau_0)=k^2/8,u_3(\tau_0)=0$. Using the preceeding result, $u_3$ is monotone and, since $u_3(0)=0$, has the same sign while $u_1>u_3$ and $u_1>0$ hold. There must therefore exist $\tau_1$ such that $u_2$ is  monotone on $[0,\tau_1]$. It follows that $u_2-k^2/8$ has the same sign on $(0,\tau_1).$ Suppose that $u_2-k^2/8>0$ on $(0,\tau_1)$. We can choose $\tau_1$ such that $R>0$ on $(0,\tau_1)$. Then, using \refb{ldot} and $R>0$, $\dot{l}=u_3$ must be negative on $(0,\tau_1)$, which is a contradiction. A similar argument rules out $u_2-k^2/8<0$ on $(0,\tau_1)$, so we have $u_2-k^2/8=0$ for all $\tau\in(0,\tau_1)$. If $u_2$ is constant on $(0,\tau_1)$ then $u_3=0$ must also hold there. \end{proof}
\begin{proposition} Recall $m=\sqrt{\lambda}/2$. If $V_0e^{\lambda l_0}=k^2/8$ and $\lambda<0$ then there is a singularity at $\tau=\infty$, which is reached by radial null rays in finite affine time. If $V_0e^{\lambda l_0}=k^2/8$ and $\lambda>0$ then there is a singularity at $\tau=\pi/m$, which is reached by radial null rays in finite affine time. \end{proposition}
\begin{proof}
Using the preceeding result, we have $u_2=k^2/8, u_3=0$, and thus $\ddot{R}=-\lambda R/4$, for all $\tau\in(0,\tau_M)$. The solutions in terms of $S$ are 
\beq
S= \left\{\begin{array}{c}m^{-1}e^{-\tau/2}\sin m\tau, \quad\hspace{7pt}\mbox{if } \quad\lambda >0.\\
    m^{-1}e^{-\tau/2}\sinh m\tau, \quad \mbox{if } \quad\lambda <0.\end{array}\right.
\eeq
Note that the case $\lambda=0, u_2(0)=k^2/8$ is precisely the case $k^2=2,V_0=1/4$ covered in proposition 5.1. If $\lambda<0$ then we clearly have $\tau_M=+\infty$. In this case we also have $m=1/2-k^2/4<1/2$ and so $\lm S=0$. Using $\dot{l}=0$,$V_0e^{\lambda l}=k^2/8$ and \refb{Ricci} the Ricci scalar reduces to 
\beq
\label{R}\mathcal{R}=\frac{3k^2e^{-k^2l_0/2+\tau/2-c_1}}{8|u|},
\eeq
and it immediately apparent that $\lm \mathcal{R}=+\infty$. \\
In the case $\lambda>0$ we have $S(\pi/m)=0$. In the cases studied thus far, surfaces characterised by $S=0$, other than the regular axis, have been singular, which was demonstrated by an infinite curvature invariant. In this case, however, it is clear from \refb{R} above that $\mathcal{R}$ is finite if $\tau$ is finite, and one can check that this is the case for other invariants such as $\mathcal{T}=T^{ab}T_{ab}$ and the Kretschmann scalar $R^{abcd}R_{abcd}$. However, the specific length of the cylinders $L$ limits to zero as $\tau\rightarrow\pi/m$, which violates the regular axis conditions. This may be seen solving \refb{EFE tau C} for $\phi$ given the solutions for $R=e^{\tau/2}S$ given above, which yields
\beq
e^\phi = \frac{e^{\phi_0 +\tau/4}}{\cos(m\tau/2)}. 
\eeq
Recalling that $L=|u|e^{-\phi}$, we have $\lim_{\tau\rightarrow\pi/m}L=0$. We speculate that we have a non-scalar curvature spacetime singularity at $\tau=\pi/m$ in this case. 
The solution to the geodesic equation \refb{geo} with $l=l_0$ shows that $\mu$ is finite for all $\tau>0$ in both cases. \\
\end{proof}
\section{Proof of Theorem 2.3}
In this section we gather the results from the two previous sections which give the proof of Theorem 2.3. \\
\\
\textbf{Proof of Theorem 2.3}
\begin{proof} The proof of cases 1 and 2 are given by Propositions 4.1 and 5.2, respectively. For case 3, part (i) is given by Proposition 4.2 and part (ii) is given by Proposition 4.3. Case 4 part (i) is given by Propositions 4.4, part (ii) is given by Propositions 4.5 and 5.4, and part (iii) by Proposition 5.3. Case 5 is proven by Propositions 5.1 and 5.4. 
\end{proof}
\section{Conclusions and further work}
We have determined the global structure of solutions in the causal past of the singularity at $\mathcal{O}$ for all values of the parameters $V_0$ and $k$ and the initial datum $l_0$. For $k^2\ge2$, the spacetime terminates either on or before the surface $\mathcal{N}_-$. For $k^2<2$, solutions exist on $\mathcal{N}_-$, which are regular, and may be extended into region \textbf{II}. In a follow up paper, we investigate the evolution of these solutions with a view to answering the question of cosmic censorship relative to this class of spacetimes.

\section{Acknowledgments}
BN acknowledges gratefully fruitful discussions with Hideki Maeda, who suggested this problem, and who shared his preliminary calculations with the authors. This project was funded by the Irish Research Council for Science, Engineering and Technology, grant number P07650.

\end{document}